\documentclass[lettersize,journal]{IEEEtran}
 \pdfoutput=1

\usepackage{balance}
\usepackage{graphicx}
\usepackage{caption}
\usepackage [noadjust]{cite} %[nospace]
\usepackage{bm}  % bold math for greek symbols
\usepackage{amsmath}
\usepackage{amssymb}
\usepackage{enumerate}
\usepackage{stfloats}
\usepackage{cases}
\usepackage{epstopdf}
\usepackage{soul,color} % \ul  \hl
\usepackage{hyperref}
\usepackage[percent]{overpic}
\usepackage{tabularx}
\usepackage[table]{xcolor}
\usepackage{multirow}
\usepackage{booktabs}  % multicolumn tabular
\usepackage{tabu} % different font sizes in different rows of tabular

\usepackage{algorithm}
\usepackage{algpseudocode}% 
\newcommand{\M}{\mathcal{M}}

\newcolumntype{L}[1]{>{\hsize=#1\hsize\raggedright\arraybackslash}X}%
\newcolumntype{R}[1]{>{\hsize=#1\hsize\raggedleft\arraybackslash}X}%
\newcolumntype{C}[1]{>{\hsize=#1\hsize\centering\arraybackslash}X}%

\newcommand{\argmin}{\operatornamewithlimits{argmin}}

\newcommand{\hbold}{\boldsymbol{h}}

\newcommand{\npower}{\sigma^2}
\newcommand{\wbold}{\boldsymbol{w}}
\newcommand{\wboldopt}{\boldsymbol{w}^{\mathrm{opt}}}
\newcommand{\wboldprimeopt}{\boldsymbol{w}^{\mathrm{'opt}}}
\newcommand{\Wspace}{\mathcal{W}}
\newcommand{\sbold}{\boldsymbol{s}}
\newcommand{\pbold}{\boldsymbol{p}}
\newcommand{\pibold}{\boldsymbol{\pi}}
\newcommand{\abold}{\boldsymbol{a}}
\newcommand{\thetabold}{\boldsymbol{\theta}}
\newcommand{\rbold}{\boldsymbol{r}}

\newtheorem{theorem}{Theorem}

\newtheorem {remark}{Remark}
\begin{document}
	\raggedbottom
	\allowdisplaybreaks
	%\title{A Novel Fresnel Zone WPT Structure using Dense Metasurfaces based on Distributed DRL Networks}
    \title{6G Fresnel Spot Beamfocusing using Large-Scale Metasurfaces: A Distributed DRL-Based Approach
    \thanks{This research was supported by Business Finland via project 6GBridge - Local 6G (Grant Number: 8002/31/2022), and Research Council of Finland, 6G Flagship Programme (Grant Number: 346208).}}
	\author{Mehdi Monemi, \textit{Member}, IEEE, Mohammad Amir Fallah, Mehdi Rasti, \textit{Senior Member}, IEEE, Matti Latva-Aho, \textit{Fellow}, IEEE 
		\thanks{
  Mehdi Monemi is with the Centre
for Wireless Communications (CWC), University of Oulu, 90570 Oulu, Finland (email: mehdi.monemi@oulu.fi).
\\
Mohammad Amir Fallah is with the Department of Engineering, Payame Noor University (PNU), Tehran, Iran (email: mfallah@pnu.ac.ir)
\\
Mehdi Rasti is double affiliated with both Centre for Wireless Communications (CWC) and the Water, Energy and Environmental Engineering research unit (WE3) at the University of Oulu, 90570 Oulu, Finland (e-mail: mehdi.rasti@oulu.fi). The work of M. Rasti was supported by the University of Oulu and the the Research Council of Finland (former Academy of Finland) Profi6 336449.
\\
Matti Latva-Aho is with Centre
for Wireless Communications (CWC), University of Oulu, 90570 Oulu, Finland
(email: matti.latva-aho@oulu.fi).
}
	}

	\maketitle
% 	\vspace{-15mm}
	\begin{abstract}
% 	In this paper, we introduce the concept of spot beamfocusing (SBF) in the Fresnel zone leveraging extremely large-scale programmable metasurfaces (ELPMs) as a key enabling technology for 6G networks. A smart SBF scheme aims to adaptively concentrate the aperture's radiating power exactly at a desired focal point (DFP) in the 3D space utilizing some Machine Learning (ML) method. This offers numerous advantages for next-generation networks including efficient wireless power transfer (WPT), interference mitigation, and improved information security. 
% SBF necessitates ELPMs with precise channel state information (CSI) for all ELPM elements. However, obtaining exact CSI for ELPMs is not feasible in all environments; we alleviate this by proposing an adaptive novel CSI-independent ML scheme based on the TD3 deep-reinforcement-learning (DRL) method. While the proposed ML-based scheme is well-suited for relatively small-size arrays, the computational complexity is unaffordable for ELPMs.
% To overcome this limitation, we introduce a modular highly scalable structure composed of multiple sub-arrays, each equipped with a TD3-DRL optimizer. This setup enables collaborative optimization of the radiated power at the DFP, significantly reducing computational complexity while enhancing learning speed. The proposed structure's benefits in terms of 3D spot-like power distribution, convergence rate, and scalability are validated through simulation results.
\textcolor{black}{
We propose a novel approach to smart spot-beamforming (SBF) in the Fresnel zone leveraging extremely large-scale programmable metasurfaces (ELPMs). A smart SBF scheme aims to adaptively concentrate the aperture's radiating power exactly at a desired focal point (DFP) in the 3D space utilizing some Machine Learning (ML) method. This offers numerous advantages for next-generation networks including ultra-high-speed wireless communication, location-based multiple access (LDMA), efficient wireless power transfer (WPT), interference mitigation, and improved information security.} 
SBF necessitates ELPMs with precise channel state information (CSI) for all ELPM elements. However, obtaining exact CSI for ELPMs is not feasible in all environments; we alleviate this by developing a novel CSI-independent ML scheme based on the TD3 deep-reinforcement-learning (DRL) method. While the proposed ML-based scheme is well-suited for relatively small-size arrays, the computational complexity is unaffordable for ELPMs.
To overcome this limitation, we introduce a modular highly scalable structure composed of multiple sub-arrays, each equipped with a TD3-DRL optimizer. This setup enables collaborative optimization of the radiated power at the DFP, significantly reducing computational complexity while enhancing learning speed. The proposed structure's benefits in terms of 3D spot-like power distribution, convergence rate, and scalability are validated through simulation results.

	\end{abstract}
	%...................................................................................................................................
	% keywords
	
	%\IEEEpeerreviewmaketitle
	
	%...................................................................................................................................
	% Introduction
\thispagestyle{empty}

\section{Introduction}

\subsection{Introduction}

Electromagnetic propagation can be investigated in three regions based on the distance from the transmitting antenna $r$, namely \textit{far-field} also known as Fraunhofer ($r>\overline{D}$), \textit{non-radiative near-field} ($r<\underline{D}$), and
\textit{radiative near-field}  ($\underline{D}<r<\overline{D}$) which is also named as \textbf{Fresnel zone}, where the values of $\underline{D}$ and $\overline{D}$ depend on several factors such as the antenna geometry, and signal wavelength \cite{selvan2017fraunhofer}. The concept of beamforming is somehow different in each of these regions.
In the far-field region, wherein the electromagnetic wave propagates in the form of plane waves, the beamforming is 2-dimensional (2D) in the direction of elevation, and azimuth angles, with no directivity dependency in the direction of $r$. In the non-radiative near-field region, where the receiving antenna is very close to the transmitting antenna, the electromagnetic propagation is dominated by the stationary inductive/capacitive field, and thus, the beamforming problem is not an issue of concern here.
In the Fresnel zone, the electromagnetic wave radiation is in the form of spherical waves enabling 3-dimensional (3D) beamforming, wherein the directivity can also be dependent on $r$. \textit{Beamfocusing} is a kind of 3D beamforming wherein the electromagnetic power is concentrated in a small region in the 3D space. In practice, however, it is of crucial importance for many applications to highly concentrate the power in a very small region (i.e., around a DFP). We call this concept \textit{spot beamfocusing} (\textbf{SBF}) in this work. 

%The necessary amplitude and phase distribution of the field imposed over the aperture can be determined in a holographic sense, by interfering with a hypothetical point source located at the DFP with a plane wave at the aperture location. 
%While conventional technologies such as phased arrays can achieve the required control over the phase and amplitude of each antenna element, they typically do so at a high cost; alternatively, metasurface apertures (PMs) can achieve dynamic focusing with potentially lower cost by only controlling the phase of each antenna elements from a quantized phase space.
%Finally, for the radiating far-field, where $d>\overline{D}$, electromagnetic energy radiates in the form of a plane wave, and thus here, instead of a 3D beamfocusing, we have a 2D beam pattern without any dependency to the distance $d$.

The sixth generation (6G) of wireless networks provides connectivity to a large number of low-cost small form-factor sensor-type Internet of Everything (IoE) devices with diverse needs enabling massive machine-type communication (mMTC)\cite{rasti2022evolution,mahmood2021machine}. SBF can play a crucial role in realizing many of the benefits of 6G IoE-mMTC networks. The massive number of connected devices in a 6G network might require a large amount of power and put a huge burden on the energy consumption of the network \cite{huang2019survey}. In line with the 6G zero-energy policy, most of the IoE devices are battery-less or have very small batteries equipped with energy harvesting technologies. A key enabling technology (KET) for optimal RF energy harvesting can be the implementation of wireless power transfer (WPT) through 3D SBF at the very exact location point of the user equipment's (UE's) receive antenna. High-capacity WPT is one of the major applications of SBF, however, there exist several other important issues wherein SBF is certainly of crucial benefit. {\color{black}For example, SBF can leverage the capacity of 6G mmWave/sub-THz network to a great extent through effective spatial frequency reuse and interference mitigation for the mMTC applications \cite{zhang20236g}  as well as realizing the location-based multiple access (LDMA) \cite{wu2023multiple}}. On the other hand, the hyper-connectivity between human and everything causes an increase in RF pollution, which is dangerous for health\cite{zhang20236g,huang2020holographic}; SBF tackles health issues by minimizing RF pollution through beam concentration. Information security is another important field wherein SBF can play a major role; more specifically, dense IoT devices with constrained resources might be highly vulnerable to malicious attacks. Here, SBF prevents the information from being captured by any device (including the attackers) positioned at any location other than that of the intended UE.

A sharp Fresnel beamfocusing inherently requires a relatively high aperture-to-wavelength ratio \cite{huggins2007introduction}, which is practically not achievable in the custom sub-6 GHz RF frequencies for indoor environments due to requiring very large-size apertures; however, with the transition of 6G systems to mmWave/THz frequencies together with the use of large-scale phased-array antennas with many antenna elements, SBF systems can practically be realized, enabling the transfer of data and energy from the aperture to the desired focal-point, whose distance can range from a few tens of centimeters to a few tens of meters. Large-scale programmable metasurfaces (ELPMs) which comprise several thousands of programmable metasurfaces (PMs), can be employed as cost-efficient large-scale apertures with the ability of smart beam control is an urgent need to implement SBF.
PM is currently one of the advanced costly-efficient architectures of the metasurfaces family, which has the excellent capability of real-time beam control. PM is composed of programmable meta-atoms, each of which has a tunable phase-shift functionality by using PIN diodes, enabling low-cost quantized phase shifts. 

For the case when the channel state information (CSI) of all antenna elements is exactly available, obtaining the optimal phase values resulting in a sharp beamfocusing is straightforward to be handled; however, the assumption of perfect near-field CSI for ELPMs is not always aligned with practical limitations in 6G networks, especially for indoor environments as justified in the following. In the classical far-field analysis of free space channel models, the propagation distance from each array antenna element to the UE is almost the same leading to the same path loss and angle of arrival for all PM antenna elements. Thus, the phase difference between the elements is solely attributable to the different spatial positioning of the elements.
However, in the near-field region, the wireless link from each array element to the UE has different path loss, and the phase variations of different antenna elements result from different positions of the array elements as well as different angles of arrival. This makes the perfect CSI estimation for all elements of the ELPMs much more complicated than that in the far-field region even for the free space model, let alone considering the multi-path channel model for indoor environments resulting from the reflections of the walls and many other reflecting objects. 
On the other hand, a sharp focal point requires that received signals from all PM elements be of the same phase; a small error in the CSI estimation for the elements results in the loss of spot beamfocusing at the desired location. Therefore, the implementation of near-field CSI-independent beamfocusing algorithms is highly preferred to those requiring exact/estimate values of the CSI. 

Based on what stated so far, and considering the very high cardinality of the set of array elements of ELPM, we are faced with a complicated problem:
\textit{{\color{black}
Presenting a scalable hardware structure and low-complexity and smart software scheme for ELPMs to realize a sharp spot-like Fresnel beamfocusing with no need for the CSI between the UE and any of the PM elements.}} As will be shown in the article, the stated problem is an NP-hard one with a very large search space which may not be directly handled with any of the traditional or machine learning (ML) solution schemes.  To tackle this issue, we employ a novel structure consisting of modular PM sub-arrays each having an independent deep reinforcement learning (DRL) optimizer; a deduction scheme finally computes the final optimal solution from the outputs of all sub-array DRL optimizers.

\subsection{Background work}
Extensive research has been conducted on far-field beamforming including transmit antenna selection, as well as analog/digital beamforming for the large-scale arrays \cite{8855807,9140420,ning2023beamforming,ning2021prospective}, however, beamforming for the near-field Fresnel zone has been less dealt with compared to that for the far-field.
{\color{black}Generally, the works on near-field beamforming are divided into two categories: those dealing with beamforming or synthesizing the near-field beam pattern for applications such as communication, 
 localization, and sensing \cite{zhang20236g,jiang2022reconfigurable,cheng2022efficient,alexandropoulos2022near,bjornson2021primer,10146329,zhang2022beam,lee2019beamforming,cui2021near}, and those proposing structures and algorithms for focusing the beam in a confined volume \cite{nepa2017near, filippou2021experimental, tofigh2014near, peng2022noncontact, boneberg2022optical, khan2020wireless, zhang2022near, demarchou2022energy,costanzo2021evolution,clerckx2018fundamentals}. 
The authors of \cite{zhang20236g} have explored the challenges and key solution schemes for near-field MIMO and massive MIMO communications in the next-generation 6G networks.
In \cite{jiang2022reconfigurable}, a RIS structure is designed that converts the received planar waves into cylindrical or spherical waves resulting in less energy leakage and higher channel capacity, and then proposes the maximum likelihood (ML) method and the focal scanning (FS) method to sense the location of the receiver. In \cite{cheng2022efficient}, an iterative low-complexity beam management algorithm is presented for solving the localization problem for multiple coherent sources in the near-field region.}
The authors of \cite{alexandropoulos2022near } have proposed a novel beam management technique that leverages the near-field interactions between RIS elements and antenna arrays, allowing for hierarchical control of the beam patterns and polarization.
As opposed to the conventional far-field beamforming, near-field beamforming has a finite depth; this is analytically investigated in \cite{ bjornson2021primer}, wherein the authors have specified a distance range for RIS-enabled near-field beamforming where finite-depth beamfocusing is possible, and the distance range where the beamforming gain tapers off.
In \cite{10146329} the authors have derived the optimal analog 3D near-field beamforming which can effectively recover the array-gain losses in the near-field region.
{\color{black}The authors of \cite{zhang2022beam}
studied the potential of near-field beamfocusing in facilitating high-rate multi-user downlink
MIMO systems through fully digital and hybrid beamforming schemes.
In addition to the conventional near-field beamforming, the concentration of signal power around a focal point through phased array beamfocusing mechanisms have been investigated in \cite{nepa2017near}}, and then the application of near-field beamfocusing has been explored in different domains including focused medical treatment \cite{filippou2021experimental,tofigh2014near}, contactless microwave inspection \cite{peng2022noncontact}, optical imaging \cite{boneberg2022optical}, wireless power transfer (WPT) or wireless power and information transfer (WPIT) \cite{khan2020wireless, zhang2022near, demarchou2022energy, costanzo2021evolution,clerckx2018fundamentals}.

Conventional beamforming for large-scale arrays can be challenging due to the large number of antenna elements involved, making it difficult to compute the optimal phase and/or amplitude relating to each antenna element. Recently, the implementation of deep learning has drawn great attraction for the beamforming of large-scale antenna arrays \cite{alkhateeb2018deep, lin2019beamforming, eappen2022deep, zhang2021reinforcement, liu2022low, wang2209extremely}. 
In \cite{alkhateeb2018deep}, the authors proposed a deep neural network (DNN)-based beamforming method for large-scale arrays. The proposed method was trained using a dataset of previous measurements to learn the optimal beamforming weights for a given set of input signals, achieving high accuracy performance with low computational complexity. A beamforming DNN-based structure is developed in
 \cite{lin2019beamforming} to maximize the spectral efficiency; the performance of the proposed scheme is then evaluated under imperfect CSI. In \cite{eappen2022deep}, a deep learning integrated reinforcement learning (DLIRL) algorithm is proposed for comprehending intelligent beamsteering for Beyond Fifth Generation (B5G) networks. The proposed scheme includes alternate path finding during path obstruction and steering the beam appropriately between the smart base station and UE.
  In  \cite{zhang2021reinforcement} the problem of beam codebook optimization for large-scale arrays is formulated as a Markov decision process, where the agent (i.e., the transmitter) learns to select the best beam codebook through interaction with the environment.

Achieving spot-like beamfocusing using low-cost quantized-phase ELPMs requires arrays with a very large number of elements and complicated algorithms. Most existing works study this problem for the case where CSI of all array elements is available through channel estimation techniques; Such information is employed in the learning process of the DNNs \cite{10044679, wei2022codebook, shen2023multi, liu2022deep}.
In \cite{10044679}, an efficient model-based deep learning algorithm is proposed for estimating the near-field wireless channels, and then the channel estimation problem is solved by applying the Learning Iterative Shrinkage and Thresholding Algorithm (LISTA).
The authors of \cite{wei2022codebook} discuss the unique challenges associated with RIS systems for near-field zones, including the required massive array size and the need for accurate phase control; then, they %propose a novel codebook design that utilizes Discrete Fourier Transform (DFT) beams and
develop a hybrid beam training method that combines analog and digital techniques.
In \cite{shen2023multi}, the authors investigate the coverage and capacity challenges associated with traditional single-beam near-field RIS systems, and then they propose a novel multi-beam scheme to improve the overall system performance.

%The common point in these works is that they all have used one of the conventional types of near-field channel model estimation methods to learn the neural network.
As mentioned earlier, a small channel estimation error for near-field ELPMs might cause a rather considerable phase error leading to a deviation in the beamfocusing process.
Therefore, CSI-independent methods result in much more accurate beamfocusing, especially in the case of SBF applications.
\textcolor{black}{Channel estimation for ELPMs and massive MIMO antennas is a demanding and challenging task for existing and advancing communication \cite{cui2022near,10044679} and localization  \cite{pan2023ris} technologies. The efficiency of such schemes relies heavily on the sparsity of the channel matrix in the angular domain. This assumption, however, only holds when the wavefronts are planar in the far-field. In this regard, a polar-domain sparse representation of the channels with a compression ratio of around 50\% for a 256-element antenna was proposed in the near-field by the authors of \cite{cui2022channel}. However, this method is only applicable to 1D linear arrays and not to the 2D extremely large-scale antenna arrays used for SBF. In the near-field, the channel coefficients matrix is not sparse due to the spherical wavefront, even in the absence of multi-path propagation. Given this, along with the extremely large number of antenna elements required for SBF, the conventional channel estimation methods for massive MIMO do not apply to ELPMs in the near-field due to the high pilot overhead and processing load. ML-based schemes could be employed to address this challenge.}
To the best of our knowledge, \cite{zhang2023deep} is the only work in the literature proposing an ML solution to the problem of near-field beamfocusing with unknown CSI; however, this work has scalability issues and only works for near-field beamfocusing using a 1D linear array antenna or antennas of few elements, and thus, the proposed structure is not suitable for the SBF; more specifically, in the aforementioned work, simulations are only achieved for a 1D linear array which leads to a fan-beamfocusing and can not concentrate the beam in a focal point; the SBF requires large-scale \textbf{2D} metasurfaces, leading to a substantial increase in time-complexity of the ML-based algorithm proposed in \cite{zhang2023deep} which is practically unaffordable.

\subsection{Motivation and Contributions}

 The major contributions of this work are as follows:
\begin{itemize}

     \item \textcolor{black} {We presented a novel ML-based hardware and software structure for the realization of smart near-field SBF by using ELPMs without requiring the CSI of antenna elements. An efficient Fresnel-zone SBF requires a planar phased array with an extremely large number of antenna elements. % whose directivity must be fine-tuned at the very exact location point of the UE's antenna. 
     Theoretically, assigning the optimal beamforming vector for the 3D spot beamfocusing requires an exact estimation of CSI of all antenna elements, however, finding the exact CSI for the case of large-scale array elements is much more challenging for the near-field region with 3D spherical wavefront compared to that in the far-field region with the 2D planar wave-front. This is mainly due to the non-sparsity of the channel gain matrix in the near-field region, which is not the case in sparse channel matrix representation for far-field massive MIMO systems. Besides, the hybrid beamforming techniques employed for massive MIMO systems do not apply to the SBF due to the low resolution of beamforming in such systems. On the other hand, for the SBF, a small deviation in the estimation of the channel states might result in the UE being in the blind zone; this, however, is not generally the case in far-field beamforming. Considering the stated issues, we have devised a CSI-independent (\textbf{CSII}) ML-SBF mechanism that fine-tunes the radiated beam to the DFP by tracking the measured power value variations without requiring any knowledge of the CSI. To the best of our knowledge, this is the first work in the literature presenting a system for realizing near-field CSI-independent spot beamfocusing.
     }
     
    \item \textcolor{black}{To build the whole SBF hardware and software structure for the ELPMs, first we present the core unit consisting of a PM planar subarray and a DRL-based algorithm for the implementation of Fresnel-zone 3D beamfocusing.} The corresponding optimizer agent is proposed to work based on a \textit{revised} version of the twin-delayed deep deterministic policy gradient (\textbf{TD3}) DRL mechanism. To find the optimal \textit{quantized} beamforming vector, we have added a quantizer as well as a $k$-nearest-neighbor ($k$nn) searching mechanism to the standard TD3-DRL structure.  The proposed  $k$nn algorithm obtains the optimal candidates in a fast manner for $r$-bit uniform quantizers. This is different from traditional $k$nn algorithms which generally obtain non-exact approximations due to searching in only a random subset of the action space.  The numerical results demonstrate the superior performance of the proposed scheme in comparison with the preceding conventional deep deterministic policy gradient (DDPG) DRL version.

    \item The proposed DRL-based algorithm can easily be implemented for small-scale PM arrays; SBF however requires ELPMs with many antenna elements, leading to unaffordable computational complexity and learning speed due to the large cardinality of the beamforming search space as well as a large number of neurons in the deep neural networks.
    {\color{black}
    In this regard, the implementation of a centralized CSI-independent structure for obtaining optimal beamfocusing vector is computationally unaffordable for ELPMs. %That is why there exists no computationally affordable solution in the literature so far dealing with CSI-independent SBF through ELPMs.
    To address this challenge, we propose a novel \textbf{distributed DRL}} structure comprising several sub-array modules each equipped with a single TD3-DRL optimizer, all working collaboratively in order to maximally focus the radiating energy at the DFP location. We will show through simulation results that the proposed modular structure reduces computational complexity and increases the learning speed to a great extent, making it quite practical and efficient.
    
\end{itemize}

\section{System Model and Problem Formulation}
\label{sec:system_model_and_problem_formulation}

\subsection{System Model}
\begin{figure}
		\centering
		\includegraphics [width=250pt]{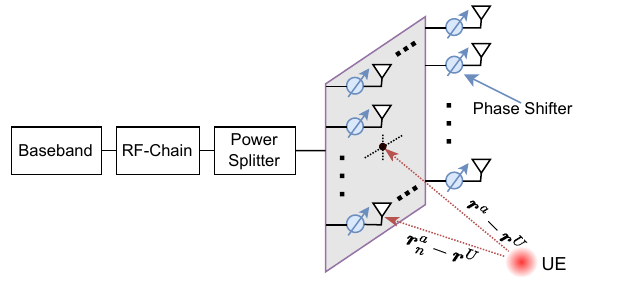} \\
% 		\vspace{-15pt}
		\caption{System model of the ELPM SBF system
		} 
		\label{fig:system_model}
% 		\vspace{-10mm}
\end{figure}
Consider a single RF-chain PM as shown in Fig. \ref{fig:system_model} consisting of a set of antenna elements $\mathcal{N}=\{1,2,..., N\}$ in a planar constellation, where $N=N_r\times N_c$ is the number of antenna elements in which $N_r$ and $N_c$ are respectively the numbers of rows and columns of the elements. {\color{black} The PMs can be of either transmissive or radiative type, and the phase of the signal emitted from/passed through each PM is considered to be 
controllable.}   The aperture diameter is denoted by $D$, and the phase of the transmitted signal through each antenna element is controlled through a quantized $r$-bit phase shifter. Let $\boldsymbol{r}_n^a$, $\boldsymbol{r}^a$, and $\boldsymbol{r}^{U}$ be respectively the location points of the antenna element $n\in\mathcal{N}$, the center of the aperture, and the UE. We consider the multipath channel gain model between antenna element $n$  and UE denoted by $h_n$ as follows:
%\begin{align}
%    h_n=k(d_{n0}^d)^{-\alpha/2}e^{j kd_{n0}^d}
%\end{align}
\begin{multline}
    \label{eq:hn}
    h_n=\eta\|\rbold_n^a-\rbold^U \|^{-\frac{\alpha}{2}}e^{-j\left( \frac{2\pi}{\lambda} \|\rbold_n^a-\rbold^U \| + \Delta\theta_{n0}\right) }
    +
    \\
    \eta\sum_{l=1}^{L} 
    \left[
    \beta_{nl} \left(d_{nl}\right)^{-\frac{\alpha}{2}} e^{-j\left(\frac{2\pi}{\lambda}d_{nl}+\Delta \theta_{nl}\right)}
    \right].
\end{multline}
The first term corresponds to the channel gain relating to the direct path from element $n$ to the UE, and the next term relates to the signals received from $L$ different reflected paths. $\eta=\left(\frac{\lambda}{4\pi}\right)^{-\frac{\alpha}{2}}$ is the attenuation coefficient in which $\lambda$ is the transmitted signal wavelength, and $\alpha$ is the path-loss exponent. $d_{nl}$ is the total path length of the propagated signal from the $l$'th path of antenna element $n$ toward the UE, {\color{black}$\beta_{nl}$ is the corresponding channel gain} relating to the signal attenuation due to the reflection (which is generally much smaller than unity), and $\Delta \theta_{nl}$ models the corresponding phase shift initiated from the reflecting surfaces in the $l$'th path, as well as the phase mismatch due to hardware impairments.
The effective channel gain vector between the PM and UE is formulated as 
\begin{align}
    \hbold=[h_n g_n^a g^{U}]^{N\times 1}
\end{align}
where $g_n^a$ is the directivity gain of the transmit antenna element $n$ (in which the mutual coupling between the antenna elements is considered as well), and $g^{U}$ is the directivity of the UE's antenna.
%As pointed out before, the stated channel parameters' values given in \eqref{eq:hn} are not required to be known in our proposed DRL scheme. The model leads to specified received power levels at the UE whose measured value variations in different time steps are sufficient for the DRL to focus the beam at the location point of the UE and maximize the power transfer. 
The UE's received signal denoted by $x$ is obtained as follows:
\begin{align}
    x =\wbold^H  \hbold s +\nu
\end{align}
 where $s$ is the input signal to the PM radiating elements (before entering phase shifters), {\color{black}$\wbold=\frac{1}{\sqrt{N}}[e^{j\phi_1}, e^{j\phi_2}, ...,e^{j\phi_N} ]^T$ is the phase-shift coefficient vector, and $\phi_i\in\phi^{\mathrm{quan}}$ (correspondingly $w_i\in\Wspace_o$, and $\wbold\in\mathcal{W}_o^N\equiv \mathcal{W}$), in which $\phi^{\mathrm{quan}}$ is the set of $2^r$ valid phases obtained from the $r$-bit quantized phase shifters}\footnote{ {\color{black} For example, for 2-bit uniform quantizer, we have $\phi^{\mathrm{quan}}=\{ 0,\pi\}$, $\Wspace_o=\{e^{j0},e^{j\pi}\}$, and $\mathcal{W}=\{e^{j0},e^{j\pi}\}^N$.}}, and $\nu$ is the additive noise.
 For any given vector $\wbold$, the received power at some location $\rbold^U$  is  obtained as follows:
 \begin{multline}
     \label{eq:power_def}
     p(\wbold,\rbold^U)=\mathbb{E} [ x x^*]=\wbold^H  \hbold \hbold^H \wbold \mathbb{E}[s s^*] +
     \wbold^H  \hbold\mathbb{E}[s \nu^*] +
     \\
     \hbold^H  \wbold \mathbb{E}[s^* \nu] + \npower_{\nu}.
 \end{multline}
 in which $\npower_{\nu}$ is the noise power.

 \subsection{Fresnel Zone Considerations}
 {\color{black}
Beamfocusing requires that the following relation holds.
\begin{align}
   \lVert \boldsymbol{r}^a-\boldsymbol{r}^{U}\rVert
    \in
    [\underline{D},D^F]
\end{align}
The maximum limit $D^F$ is the Fraunhofer limit which is the boundary between the near-field and far-field, and whose value is obtained as $D^F= 2D^2 /\lambda$ on the boresight of the antenna. For the off-boresight scenario where the angle between the antenna plane and the line connecting the DFP and antenna center is $\theta$, this can be obtained as $D^F= 2D^2\sin^2 \theta/\lambda$  \cite{balanis2016antenna}. The lower bound $\underline{D}$ is the boundary between non-radiative and radiative near-field regions wherein the amplitude of the radiative and reactive powers are the same, and beyond which the reactive power fades rapidly. The value of $\underline{D}$ depends on the geometry of the antenna and is generally lower than a wavelength. If the lower bound is violated, the UE is so close to the aperture and the major received power is mostly captured through the inductive/capacitive field in the non-radiating near-field region.
  On the other hand, the violation of the upper bound inequality results in the UE lying in the far-field region wherein the focal concentration of the beam is not possible. 
  }

Fig. \ref{fig:Freznel_Plots} illustrates how the location of the DFP and the size of the aperture affect the achieved focal concentration of the beam.{\color{black} While it is seen that a desired 3D focal point is not possible to be achieved when the DFP is very close to and far enough from the aperture corresponding to the non-radiating-near-field (Fig. \ref{fig:Freznel_Plots}-a) and far-field zone (Fig. \ref{fig:Freznel_Plots}-b) respectively, moderate and sharp spot focal points can be obtained in the Fresnel zone as seen in Figs. \ref{fig:Freznel_Plots}-c and \ref{fig:Freznel_Plots}-d respectively. The SBF corresponding to Fig. \ref{fig:Freznel_Plots}-d is obtained when $D/\lambda$ is very high, which requires ELPMs rather than small-scale PMs.
}
{\color{black}
\begin{remark}
    It should be noted that a sharper focal point necessitates a more dominant near-field effect and requires a closer distance between the DFP and the aperture \cite{smith2017analysis}. This means that SBF is not feasible at distances close to $D^F$,  
   and thus, the feasible region for SBF is considered as
   \begin{align}
   \lVert \boldsymbol{r}^a-\boldsymbol{r}^{U}\rVert
    \in
    [\underline{D},\overline{D}]
\end{align}
where $\overline{D}$ is some constant that holds in $\overline{D}<D^F$ and its value depends on the antenna design and the desired BFR.
\end{remark}
\begin{remark}
To achieve near-field SBF with a lower BFR, or at higher distances of the focal point from the antenna, a higher value of the ratio $D/\lambda$ is required \cite{goodman2017introduction}, which can be accomplished through decreasing $\lambda$ as well as increasing $D$. The former is achieved through transitioning to higher frequencies (mm-wave, sub-THz, and THz). The latter can be realized by increasing the number of array elements, leading to larger-scale and more costly antennas.
\end{remark}
}

% \begin{figure}[t]
%   \centering
%   \subfigure{
%     \includegraphics[width=160pt\linewidth]{pictures/simulation/Fresnel/1/3D.pdf}
%   }
%   \subfigure{
%     \includegraphics[width=160pt\linewidth]{pictures/simulation/Fresnel/2/3D.pdf}
%   }
%   \subfigure{
%     \includegraphics[width=160pt\linewidth]{pictures/simulation/Fresnel/3/3D.pdf}
%   }
%   \\
%   \subfigure{
%     \includegraphics[width=160pt\linewidth]{pictures/simulation/Fresnel/1/2D.pdf}
%   }
%   \subfigure{
%     \includegraphics[width=160pt\linewidth]{pictures/simulation/Fresnel/2/2D.pdf}
%   }
%   \subfigure{
%     \includegraphics[width=160pt\linewidth]{pictures/simulation/Fresnel/3/2D.pdf}
%   }

%   \caption{Overall caption for all subfigures}
%   \label{fig:subfigures}
% \end{figure}

% \begin{figure}[t]
%   \begin{subfigure}
%     \includegraphics[width=0.25\linewidth]{pictures/simulation/FresnelC/1/3D.pdf}
%   \end{subfigure}\!\!\!\!\!\!
%   \begin{subfigure}
%     \includegraphics[width=0.25\linewidth]{pictures/simulation/FresnelC/2/3D.pdf}
%   \end{subfigure}\!\!\!\!\!\!
%   \begin{subfigure}
%     \includegraphics[width=0.25\linewidth]{pictures/simulation/FresnelC/3/3D.pdf}
%   \end{subfigure}\!\!\!\!\!\!
%   \begin{subfigure}
%     \includegraphics[width=0.25\linewidth]{pictures/simulation/FresnelC/4/60x60_target.pdf}
%   \end{subfigure}
%   \\
%   \vspace{-5pt}
%   \footnotesize{
% \begin{tabularx}{\textwidth}{ C{1} 
% C{1}  C{1}  C{1} }
%     (a) & (b) & (c) & (d)
% \end{tabularx}
% }
\begin{figure}[t]
    \begin{tabular}{c c} 
    \hspace{-10pt} 
    \includegraphics[width=0.5\linewidth]{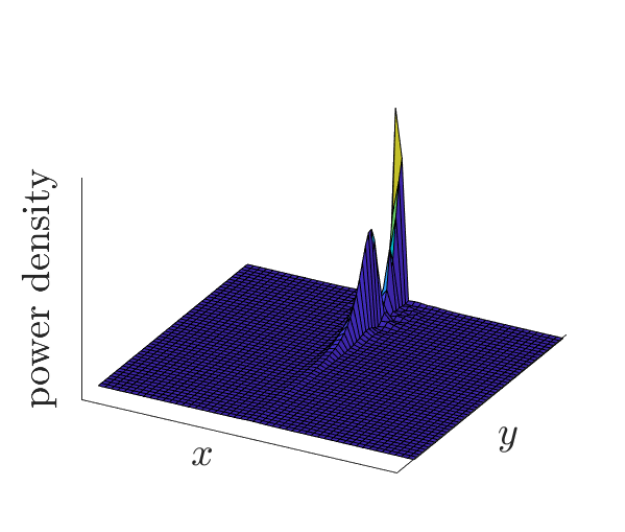}
  &
   \hspace{-10pt} \includegraphics[width=0.5\linewidth]{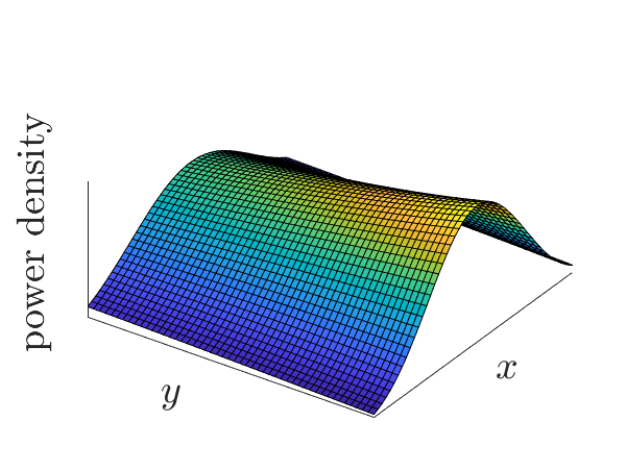}
  \\
    (a) & (b)
\\
    \hspace{-10pt} 
    \includegraphics[width=0.5\linewidth]{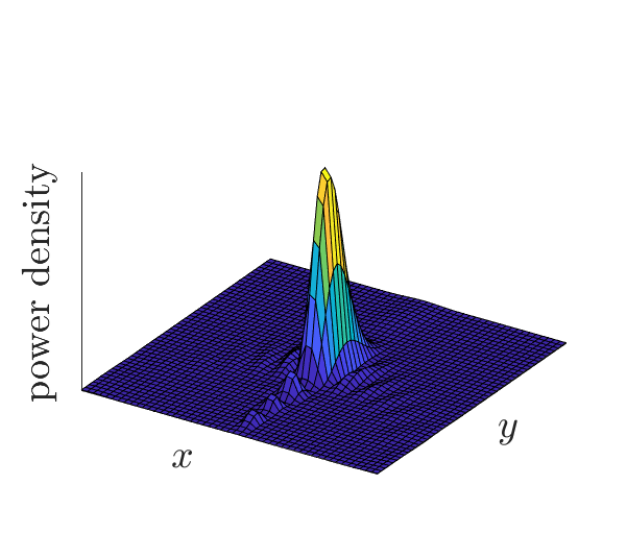}
    &
        \hspace{-10pt} 
    \includegraphics[width=0.5\linewidth]{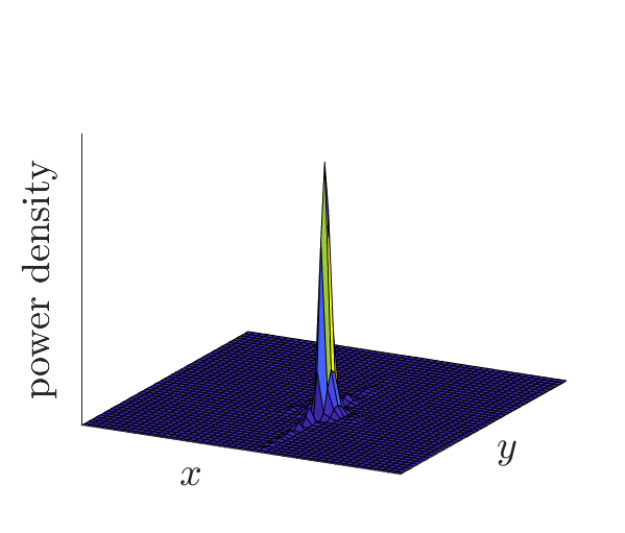}
  \\
    (c) & (d)
\end{tabular}
  \caption{Different scenarios for beamfocusing wherein an aperture is located on the $xz$ plane and the focal point is located on the $xy$ plane in front of the aperture: {\color{black}(a) DFP located in the non-radiating near-field region. %, and DFP at $y=0.3$m. 
  (b) DFP located in the far-field region. 
  (c) DFP located in the Fresnel region using a small-scale PM. %, and focal point at $y=1$m. 
  (d) SBF at the DFP realized in the Fresnel region using ELPM.%, and SBF at $y=1$m.
  }}
  \label{fig:Freznel_Plots}
\end{figure}

 \subsection{Problem Formulation}
\begin{figure}
		\centering
		\includegraphics [width=120pt]{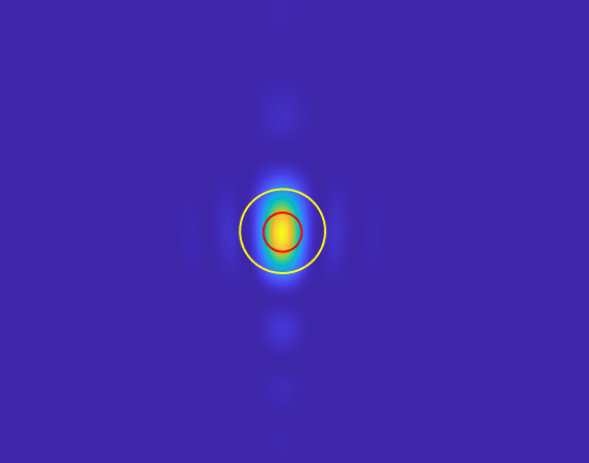} \\
% 		\vspace{-15pt}
		\caption{{\color{black}Two BFRs corresponding to two different values of $\eta$. The yellow and red circles respectively have BFRs $R_1$ and $R_2$, containing  $\eta_1=90$\% and $\eta_2=$50\% of the total power in the reference plane.}
		} 
		\label{fig:BFR}
% 		\vspace{-10mm}
\end{figure}
{\color{black}
    In what follows, first, we formally define the concept of beam focus radius (BFR) and then express the problem formulation.
    
    Let define $R(\wbold,\rbold^U,\eta)$ corresponding to the beamforming vector $\wbold$,  DFP location $\rbold^U$, and some given constant $0<\eta< 1$, as the radius of the circle $S_R$ located on the reference plane $S$ and centered at DFP, through which a fraction $\eta$ of the total radiating power in the reference plane $S$ is passed.
    For example, a BFR value corresponding to $\eta=0.9$ implies that a circle of radius BFR centered at DFP contains 90\% of the total radiating power at the reference plane.  i.e., $R(\wbold,\rbold^U,\eta)\equiv R$ is formally obtained by finding a value of $R$ (corresponding to $S_R$), for which the following  equality holds: 
    \begin{align}
    \label{eq:R}
        \int_{S_R} \partial \pbold(\wbold,\rbold').\widehat{\boldsymbol{a}}_nd s'
            = \eta 
            \int_{S^c} \partial \pbold(\wbold,\rbold').\widehat{\boldsymbol{a}}_nd s'=\eta P^T
    \end{align}
    %where $P^T$ is the total transmit power, and
    where $\partial\pbold(\wbold,\rbold')$ is the power density function corresponding to $\wbold$ at location point $\rbold'$, and $\widehat{\boldsymbol{a}}_n$ is the unit vector normal to the reference surface (focal plane), and $P^T$ is the total power passing from the focal plane. Considering a sphere of radius $R$ at the DFP, as will be shown later, when beamfocusing is achieved through ELPMs in the near-field dominant region, the power level fades outside the sphere in all directions, even when getting close to the aperture. This implies that the focal region in the near-field can be imagined as a sphere centered at the DFP. Fig. \ref{fig:BFR} depicts two circles around the focal point corresponding to two values of BFRs. %, and whose radius is related to $R$. %The radius of this sphere might be somewhat higher than $R$ since the rotation of the reference plane difference reference pl
  }
    
 Given a large-scale PM with a DFP at $\rbold^U$, the optimal beamfocusing problem 
 is formally defined as finding the beamforming vector corresponding to the minimum BFR under the physical 
 Fresnel constraint (phy. const.)  as follows:
\begin{subequations}
\label{eq:opt1}
    \begin{align}
        \mathrm{\textbf{P1:}}
        \hspace{30pt}
	\min_{\wbold} & \ R(\wbold,\rbold^U,\eta) 
        \\
        \mathrm{subject\ to:} & \ \wbold \in \mathcal{W} 
        \\
        \mathrm{phy.\ const.:} & \ \lVert \boldsymbol{r}^a-\boldsymbol{r}^{U}\rVert\in [\underline{D},\overline{D}]
    \end{align}
\end{subequations}
 On the other hand, we may define the optimal beamfocusing vector as the one corresponding to the maximum WPT at the desired point $\rbold^U$, which is stated as follows:
     \begin{subequations}
     \label{eq:opt2}
	\begin{align}
        \mathrm{\textbf{P2:}}
        \hspace{30pt}
	    \max_{\wbold} & \ p(\wbold,\rbold^U) 
        \\
        \label{eq:w2}
        \mathrm{subject\ to:} & \ \wbold \in \mathcal{W} 
        \\
        \mathrm{phy.\ const.:} & \ \lVert \boldsymbol{r}^a-\boldsymbol{r}^{U}\rVert\in [\underline{D},\overline{D}]
	\end{align}
 \end{subequations}
 If we consider exact CSI for all antenna elements is available, and a continuous phase beamforming vector is allowed (i.e., constraint \eqref{eq:w2} is relaxed), problem {\bf P2} is convex as will be later shown in this section. However, even with these simplifying assumptions, problem {\bf P1} is not a well-behaved problem because there exists no closed-form solution to obtain $R(\wbold,\rbold^U,\eta)$ from the integral equation \eqref{eq:R}. The problems \textbf{P1} and \textbf{P2} for ELPMs in the Fresnel region are equivalent in many scenarios. In what follows we show that for the focal reference plane, one can solve the more straightforward problem \textbf{P2} instead of the original SBF problem \textbf{P1} for extremely large-scale PMs.
 {\color{black}
     \begin{theorem}
     \label{clm:1}
        %The optimal SBF problem \textbf{P1} and WPT problem \textbf{P2} for ELPMs are equivalent, in the sense that they both lead to the same optimal power $p^*$ and BFR $R^*$.
        Let $\wbold_1^*$ and $\wbold_2^*$ be the solutions to {\bf P1} and {\bf P2} respectively. The optimal SBF problem \textbf{P1} and WPT problem \textbf{P2} for ELPMs are equivalent for the focal reference plane, in the sense that $R(\wbold_2^*,\rbold^U,\eta)= R(\wbold_1^*,\rbold^U,\eta)$ and $p(\wbold_2^*,\rbold^U)= p(\wbold_1^*,\rbold^U)$, provided that the near-field effect is dominant and the array neighboring elements spacing is not higher than half a wavelength.
       
    \end{theorem}
     \begin{proof}
        %The claim is proved by contradiction.
        First note that due to the assumption of interelement spacing not higher than a half wavelength, there exists a single main lobe (i.e.,  no grating lobes exist). This means that there only exists one focal region. Besides, the wave propagation of an aperture in the Fresnel zone can be well approximated as a Gaussian function around the focal point in the focal reference plane \cite{goodman2017introduction}. 
        Let $R_1=R(\wbold_1^*,\rbold^U,\eta)$. Considering problem {\bf P1} , from \eqref{eq:R},  we have
        \begin{align}
                   \label{eq:6454}
            &\int_{S_{R_1}} \partial \pbold(\wbold_1^*,\rbold').\widehat{\boldsymbol{a}}_nd s'
            %=\eta 
            %\int_{S^c} \partial \pbold(\wbold_2^*,\rbold').\widehat{\boldsymbol{a}}_nd s' 
            =\eta  P^T
        \end{align}
        For problem {\bf P2}, we consider a BFR corresponding to $\wbold_2^*$ and the same value of $\eta$ denoted by $R_2=R(\wbold_2^*,\rbold^U,\eta)$. Therefore we have
            \begin{align}
                 \label{eq:6453}
            &
            \int_{S_{R_2}} \partial \pbold(\wbold_2^*,\rbold').\widehat{\boldsymbol{a}}_nd s'
            %=\eta 
            %\int_{S^c} \partial \pbold(\wbold_1^*,\rbold').\widehat{\boldsymbol{a}}_nd s
            = \eta P^T
        \end{align}
        %The last equality is obtained from the non-sparsity of the aperture leading to the fact that there exists no grating lobes, as well as the assumption that the power radiated due to side lobes is negligible compared to the main lobe. 
        %From \eqref{eq:6452} we can write the following equality 
        which leads to
         \begin{align}
            \label{eq:p12}
            \int_{S_{ R_1}} \partial \pbold(\wbold_1^*,\rbold').\widehat{\boldsymbol{a}}_nd s'
            = 
            \int_{S_{R_2}} \partial \pbold(\wbold_2^*,\rbold').\widehat{\boldsymbol{a}}_nd s'
        \end{align}
        %Let $\wbold_1^*$ and $\wbold_2^*$ be the solutions to {\bf P1} and {\bf P2} respectively.
        On the other hand, from the definition of {\bf P1}, we have  
        %\begin{subequations}
%\label{eq:P1_infered_ineqalities}
        \begin{align}
            \label{P1_infered_ineqalities_a}
            %R(\wbold_2^*,\rbold^U,\eta)\geq R(\wbold_1^*,\rbold^U,\eta)
            R_2 \geq R_1
            %\\      
            %\label{P1_infered_ineqalities_b}  p(\wbold_2^*,\rbold^U)\geq p(\wbold_1^*,\rbold^U)
        \end{align}
        %\end{subequations}
        %Considering a rather small size antenna at the focal point where the power density is nearly constant over the surface of the antenna, \eqref{P1_infered_ineqalities_b} results in
        Considering that {\bf P2} corresponds to the maximum power level at the focal point, we have
        \begin{align}
        \label{P1_infered_ineqalities_c}
             \partial \pbold(\wbold_2^*,\rbold^U).\widehat{\boldsymbol{a}}_n\geq \partial \pbold(\wbold_1^*,\rbold^U).\widehat{\boldsymbol{a}}_n
        \end{align}
         From this, together with the Gaussian-like beams of       $p(\wbold_2^*,\rbold')$ and $p(\wbold_1^*,\rbold')$ in the focal plane, one can infer that if equality does not hold for 
        either of  \eqref{P1_infered_ineqalities_a} or \eqref{P1_infered_ineqalities_c}, the right hand of \eqref{eq:p12} becomes greater than the left hand. This completes the proof.
        %Considering this together with Gaussian-like beams of        $p(\wbold_2^*,\rbold')$ and $p(\wbold_1^*,\rbold')$, we infer that equalities must hold for both relations of \eqref{eq:P1_infered_ineqalities} as    $R(\wbold_2^*,\rbold^U)= R(\wbold_1^*,\rbold^U)$ and        $p(\wbold_2^*,\rbold^U)= p(\wbold_1^*,\rbold^U)$
     \end{proof}
}

 Based on Theorem \ref{clm:1}, we can interchangeably deal with problem \textbf{P2} instead of the original problem \textbf{P1}. From \eqref{eq:power_def}, since the signal and noise are uncorrelated, the optimization problem \textbf{P2} is written as:
 \begin{subequations}
     \label{eq:opt22}
	\begin{align}
	    \max_{\wbold} & \ \wbold^H \boldsymbol{Q} \wbold
        \\
        \label{eq:opt2_const_codebook}
        \mathrm{subject\ to:} & \ \wbold \in \mathcal{W} \\
        \mathrm{phy.\ const.:} & \ \lVert \boldsymbol{r}^a-\boldsymbol{r}^{U}\rVert\in [\underline{D},\overline{D}]
	\end{align}
 \end{subequations}
 where $\boldsymbol{Q}=\hbold \hbold^H$. 
 The optimization problem \eqref{eq:opt22} requires exact estimation of the channel gains $\hbold$ for all of the array elements. If $\hbold$ is exactly estimated, and we allow continuous phase values\footnote{For the case of exact CSI estimation and quantized phase according to constraint \eqref{eq:opt2_const_codebook}, a straightforward solution scheme is to first quantize the phase space corresponding to the solution of the non-quantized convex problem (i.e., relaxing constraint \eqref{eq:opt2_const_codebook}), and then search among the $k$ nearest neighbors to find the one corresponding to the highest objective value.}  for $\mathcal{W}$, \eqref{eq:opt2} is a convex quadratic optimization problem because $\boldsymbol{Q}$ is a positive definite matrix.
 However, as mentioned earlier, the exact estimation of the channel gains of ELPMs in the 3D-wavefront Fresnel zone is a challenging issue compared to that of the 2D-wavefront far-field region, and besides, a tiny CSI estimation error for some of the array elements results in the UE to lie in the blind zone of the radiated beam. Considering these practical limitations, instead of dealing with problem \eqref{eq:opt2} which requires exact estimation of the channel gains of all antenna elements, we consider the original problem \textbf{P2}, and propose a novel DRL scheme for finding the optimal beamfocusing vector through measuring the available received power $p(\wbold,\rbold^U)$ without requiring the exact or estimated knowledge of $\hbold$.
 
 %Generally, the exact estimation of the channel gains of dense PMs in the Fresnel zone is a challenging issue compared to that in the far-field due to the following reasons: Firstly, the directivity gain relative to isotropic antenna for a phased array in the far-field planar wavefront is only a function of the 2-dimensional (2D) azimuth and elevation angles whose estimation is less challenging compared to that of the Fresnel-zone with the spherical wavefront which is affected by the third radius dimension as well in a 3D space. Secondly, a small deviation in the estimation of the channel states in the far field scenario seldom results in the UE to lie outside the main lobe of the directivity pattern, however, this is not the case for the Fresnel-zone spot beamfocusing wherein a tiny deflection of the focused beam would result in the UE to be in the blind zone of the radiated beam. Considering these practical limitations, instead of dealing with problem \eqref{eq:opt2} which requires exact estimation of the channel gains of all antenna elements, we consider the original problem \eqref{eq:opt2}, and propose a novel DRL scheme for finding the optimal beamfocusing codebook through measuring the available received power $p(\wbold)$ without knowing or estimating any of the channel gains of $\hbold$.

\section{Proposed Solution Scheme}
As stated before, we are pursuing a CSI-independent solution scheme to the stated SBF problem, which explicitly implies that conventional convex/non-convex solution schemes are not applicable here. DRL has proven to be one of the best ML schemes to learn the solution to complex optimization problems. An efficient DRL-based SBF mechanism for ELPMs requires a large number of antenna elements each having a $2^r$ bit quantized-phase domain. This results in a very large action space, which in turn requires a very high number of neurons in the hidden layers of deep NNs leading to the extremely high computational complexity of the DRL algorithm. For example, for a $60\times 60$ 4-bit ELPM, the action space consists of $16^{3600}$ vectors which is extremely large and can not be directly handled through conventional ML methods.
On the other hand, a sharp SBF requires the DRL to tune the concentrated power in a very small spot-like zone in the 3D geometrical space using a 2D/3D antenna array structure. In other words, a 1D linear antenna array can only focus the beam on a line with infinite points (fan-beam), and thus it may not achieve the SBF; this also reveals the fact that by even considering the same number of array elements, the convergence rate of a DRL algorithm for a 2D array is much lower than that of a 1D array. This is why no work in the literature exists dealing with  ML-based CSI-independent SBF for 2D large-scale arrays. 

To tackle these issues, this section is organized as follows. First, we propose the overall low-complexity distributed structure for the SBF consisting of many sub-array modules, as presented in subsection A. For each sub-array module, instead of the existing discrete ML schemes such as deep Q network (DQN), double DQN, and DuelDQN, which all have serious scalability issues for searching in such extreme action-space domains, we use a more scalable approach by employing a quantized revision of the original continuous-action-space TD3-DDPG DRL scheme. This is fully described in subsection B. Finally, the obtained results for all sub-array modules are employed to form the overall solution as discussed in subsection C.

\subsection{The overall structure of the proposed system}

\begin{figure}
		\centering
		\includegraphics [width=254pt]{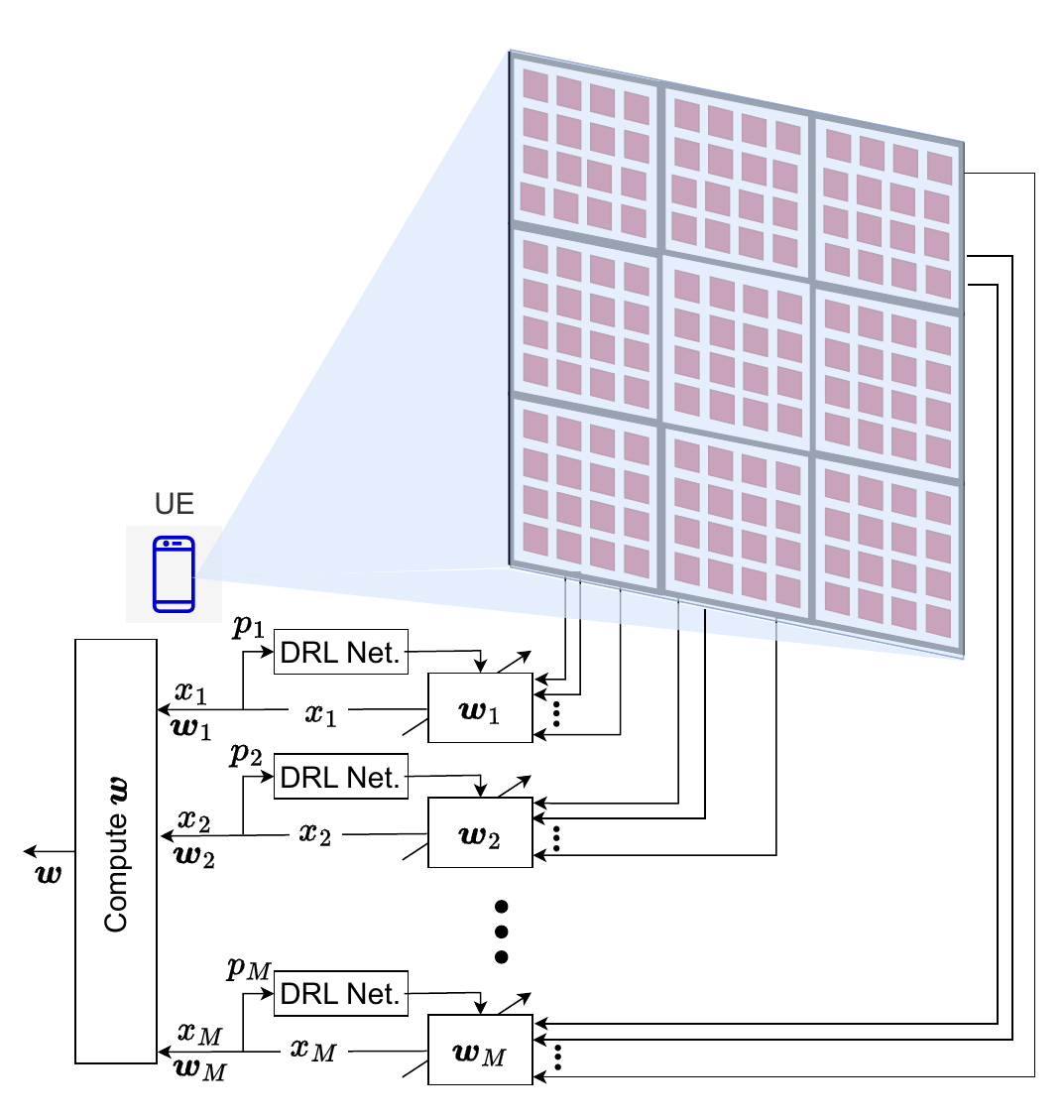} \\
% 		\vspace{-15pt}
		\caption{Hardware structure of the proposed modular SBF system for a sample case consisting of $M=9$ modules.
		} 
		\label{fig:structure_hardware}
% 		\vspace{-10mm}
\end{figure}
Consider a modular array structure wherein the antenna is comprised of $M$ similar sub-arrays each having $N'=N/M$ radiating elements. %The sub-arrays are designed in a way that the UE is potentially located in the Fresnel zone of all modules.  
An example of such a modular array system comprising $M=9$ modules is depicted in Fig. \ref{fig:structure_hardware}.
%As seen in the figure, the signal output of each sub-array module $m$ is tuned 
    Let $m\in\mathcal{M}$ be a given  sub-array module index where $\mathcal{M}=\{1,2,...,M\}$, $\wbold_m$ be the phase shift coefficient vector of the $m$'th  module, and
     $w_{mi}$ be the $i$'th element of $\wbold_m$.
     Each subarray $m$ is equipped with a DRL for tuning $\wbold_m$. For each iteration of the learning process of each module $m$, the corresponding DRL agent applies the vector $\wbold_m$, resulting in the signal $x_m$, and uses the measured power $p_m=|x_m|^2$ to update $\wbold_m$ for the next time step.
 Without loss of generality we consider that $\wbold_m$ and $\wbold$ are simply related as follows:
     \begin{align}
         \wbold=[ \underbrace{w_1,...,w_{N'}}_{\wbold_1},
         \underbrace{w_{N'+1},...,w_{2N'}}_{\wbold_2},
         ...,
         \underbrace{w_{MN'-N'+1},...,w_{MN'}}_{\wbold_M}
         ]^T.
     \end{align}
     %As long as the diameter of all modules is designed in a way that the Fresnel zone holds for the UE, 
     Instead of directly finding the optimal high dimensional vector $\wbold$ by using the original high-complexity problem \textbf{P2}, we first obtain the solution to the set of following lower dimensional optimization problems for each sub-array module $m$.
     \begin{subequations}
     \label{eq:optm}
	\begin{align}
	    \max_{\wbold_m} & \ p_m(\wbold_m,\rbold^U) 
        \\
        \mathrm{subject\ to:} & \ \wbold_m \in \mathcal{W}^{N'} \\
        \label{eq:optm_fres}
        \mathrm{phy.\ const.:} & \
        \lVert \rbold^{sub}_m-\boldsymbol{r}^{U}\rVert \in [\underline{D}_0, \overline{D}]
	\end{align}
 \end{subequations}
  where  $p_m(\wbold_m,\rbold^U)$ is the UE received power at $\rbold^U$ radiated from sub-array $m$ corresponding to $\wbold_m$, $\rbold^{sub}_m$ is the center of the sub-array aperture $m$, and $\underline{D}_0$ is the Fresnel limit of each sub-array aperture. Once the optimal solution for each sub-array module is found, the optimal solution of the original problem \textbf{P2} can be obtained in a scheme explained in section \ref{sec:overall_codebook}.
 \begin{remark}
     Please note that in \eqref{eq:optm_fres}, the Fresnel zone is estimated as
     \begin{align}
        \label{eq:region}
         [\underline{D}_0, \overline{D}]=\underbrace{[\underline{D}_0, \underline{D}]}_{A}
         \cup
         \underbrace{[\underline{D}, \overline{D}]}_{B}
     \end{align}
     As will be explained later, all sub-array PMs will finally constitute a unified ELPM with diameter $D$ and Fresnel zone $B=[\underline{D}, \overline{D}]$, %and thus it is finally important that UE be located at the distance $[\underline{D},\overline{D}]$.
     however, we have extended this region to $A \cup B$ leading to the coverage of a larger 3D area for valid Fresnel SBF.
     For the case where the UE is located inside the region $A=[\underline{D}_0, \underline{D}]$, a subset of PM modules can be deactivated to decrease the effective diameter of the ELPM, such that the UE lies within the Fresnel zone.
 \end{remark}

\begin{remark}
    For the case of the WPT application, if the UE is so close to the ELPM, due to the non-radiating near-field power induction, the maximal power transfer might not correspond to a sharp focal point, and thus, if the application is only concerned with maximum power transfer, we may assume $\underline{D}_0=0$ in \eqref{eq:region}. Similarly,  for the far-field maximum directivity application, we may assume $\overline{D}=\infty$ in \eqref{eq:region}. Therefore, \textit{if the SBF is not of concern, our proposed low-complexity DRL-based structure can apply to WPT, beamforming, and power beaming for any of the near-field and far-field cases.}
\end{remark}

\subsection{Sub-array beamfocusing through TD3-DRL scheme}
\begin{figure*}
		\centering
		\includegraphics [width=344pt]{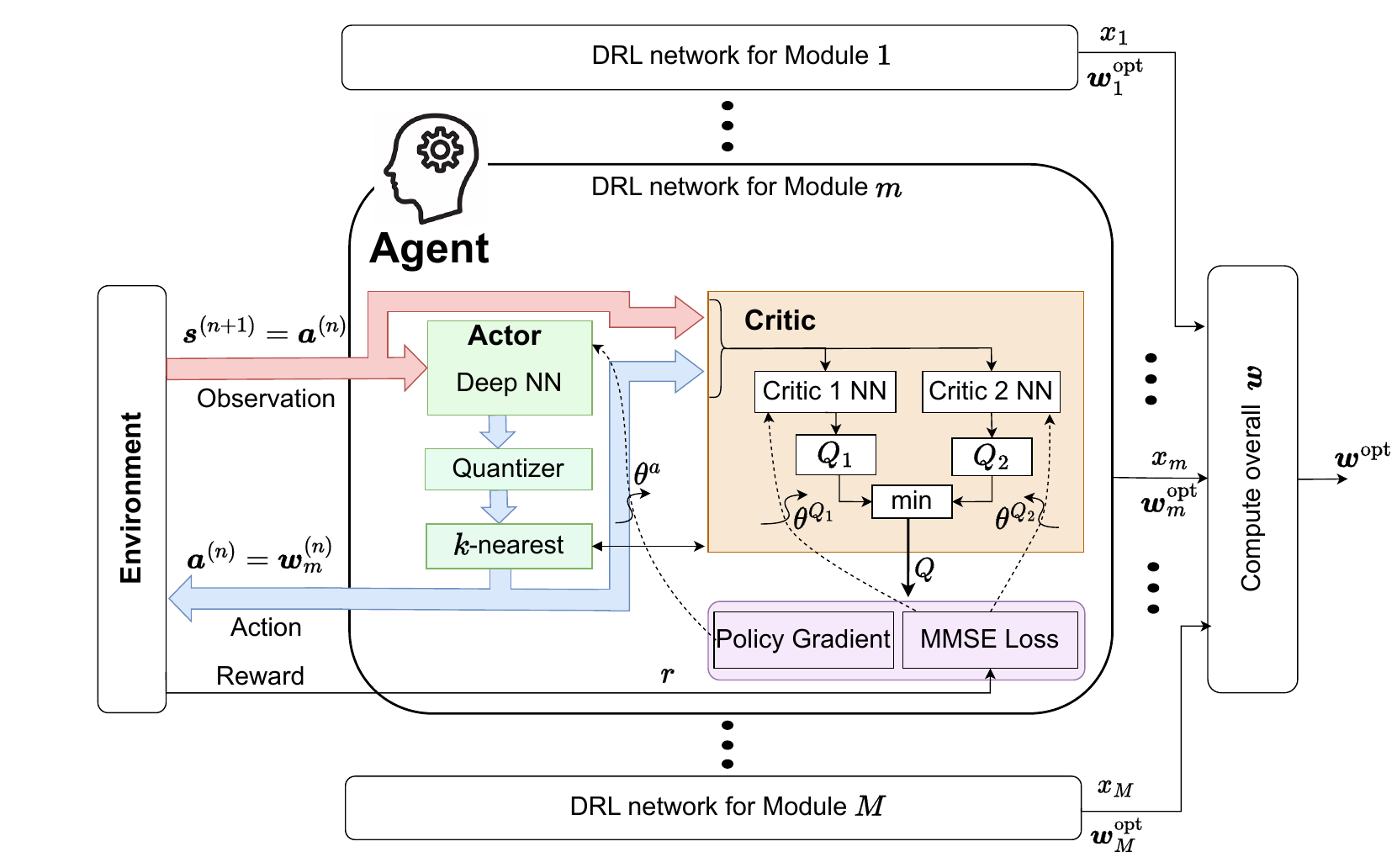} \\
% 		\vspace{-15pt}
		\caption{TD3-DRL based software structure for the proposed modular Fresnel zone SBF system showing the process of training sub-array modules and the determination of the final beamforming vector after convergence.
		} 
		\label{fig:structure_drl}
% 		\vspace{-10mm}
\end{figure*}
In this part, we propose an efficient CSI-independent ML-based beamfocusing scheme for each module $m$ to obtain the solution of \eqref{eq:optm}. The proposed algorithm is a revised version of the  TD3 scheme. TD3 is a novel model-free, online, off-policy RL method that is designed based on the actor-critic model \cite{fujimoto2018addressing}. The original TD3-DRL algorithm is designed for continuous action spaces. However, since the PM phase shifters in our problem only support discrete quantized phase values, we have revised the TD3-DRL scheme to accommodate the discrete action space as will be explained later in this section. Fig. \ref{fig:structure_drl} depicts the overall beamfocusing structure for an ELPM consisting of $M$ PM modules, wherein the detailed structure of the DRL network of module $m$ is illustrated. In what follows we explain the structure of the TD3-DRL for module $m$, as well as the overall optimal beamforming vector determination scheme in more detail. As seen in Fig. \ref{fig:structure_drl}, the TD3-DRL network for each module $m$ directly interacts with the \textit{environment} and consists of three elements of \textit{actor}, \textit{critic}, and \textit{agent}. Next, we explain each of these items.

\subsubsection{Environment} For each training step $n$, the environment receives the \textbf{action} $\boldsymbol{a}^{(n)}$ from the actor-network. Any action taken, results in some specific observation as a \textbf{state} vector denoted by $\boldsymbol{s}^{(n)}$, and a scalar \textbf{reward} denoted by $r^{(n)}$. In our problem, the action is the assigned beamforming vector for sub-array module $m$ denoted by $\boldsymbol{a}^{(n)}=\wbold_m^{(n)}\in\mathcal{W}^{N'}$. We consider that the state for each time step $n$ is the action at the previous time step, therefore, we have $\boldsymbol{s}^{(n+1)}=\boldsymbol{a}^{(n)}$. The objective is to find a beamforming vector leading to the highest power value at $\rbold^U$. For each time step $n$, the environment measures the power $p_m(\wbold_m,\rbold^U)$. We assign a unity reward for the case when the measured power of the current time step is increasing compared to the previous step, and otherwise, the reward is considered as -1, i.e.,
\begin{numcases}{r^{(n)}=}
\label{eq:reward}
+1, & if $p_m(\wbold_m^{(n)},\rbold^U)>p_m(\wbold_m^{(n-1)},\rbold^U)$
\\
-1, & otherwise 
\nonumber
\end{numcases}
\subsubsection{Actor} The actor is a DNN with network parameters $\boldsymbol{\theta}^a$ whose input is the environment state vector $\boldsymbol{s}^{(n)}\in\mathcal{S}$, and its output is the action vector  $\mathbf{a}^{(n)}\in\mathcal{A}$ using the policy $\boldsymbol{\pi}^a$:
\begin{align}
    \boldsymbol{\pi}^a(.|\thetabold^a): \mathcal{S} \rightarrow \mathcal{A}
\end{align}
For each sub-array $m$ at each iteration $n$, the input of the actor-network is simply a beamforming vector at the corresponding iteration; therefore $\mathcal{A}=\mathcal{W}^{N'}$. The actor-network in a standard TD3 structure has a continuous action space; in our problem, this corresponds to $N'$ complex numbers, each having continuous phase domain $[0,2\pi]$.
Based on what stated so far, for each time step $n$, the actor in our problem simply gets the current vector $\wbold_m^{(n)}$ and decides on the next step vector $\wbold_m^{(n+1)}$. Thus, as seen in Fig. \ref{fig:structure_drl}, we need a \textbf{quantizer} to map the continuous output beamforming vector to the nearest valid quantized vector from the space $\mathcal{W}^{N'}$. %Hence, the actor network is followed by the quantizer as follows:
%\begin{align}
%    \boldsymbol{f}^{Quan}(\pibold^a(.|\thetabold^a)): \mathbb{R}^{N'} \rightarrow \mathcal{W}^{N'}
%\end{align}
On the other hand, due to the quantization error, it is possible that the resulting quantized action does not exactly correspond to the optimal action. Therefore, similar to \cite{dulac2015deep}, we also search through the $k$ nearest neighbors (\textbf{knn}) and select the one corresponding to the highest Q-value of the critic network. An explanation of the Q value relating to the critic network will be presented later in this section.

\textbf{Proposed knn Algorithm:}
The $k$ nearest neighbors are the first $k$ actions $\widehat{\wbold}_m$ whose $\mathrm{L}2$ distance norm are closest to the desired action $\wbold_m$. The existing $k$nn algorithms for large domain spaces are non-exact approximations due to searching in only a random subset of the domain space.
For example, an antenna array of 64 elements with 4-bit digital phase shifters has $16^{64}$ states to be searched which is computationally unaffordable.  In practice, however, an $r$-bit quantizer generally divides the phase space  $[0,2\pi]$ into $2^r$ equally spaced levels. For this case, we devise the very low-complexity \textbf{Algorithm 1} to exactly find the knn candidates. As seen in Algorithm 1, the quantization difference level variable $L$ is initially set to unity in line 4 (meaning that the initial knn vectors are different with the original input vector in only one bit), and then the tensor $b$ is filled with one vector as $b_1=[1,...,N',-1,...,-N']$. For each iteration, one of the elements of $b_1$ is randomly selected, and then the corresponding element $i$ of $\widehat{\wbold}_m$ is changed by one quantization level upward in line 15 or downward in line 17 depending on the sign of the selected random number. Finally, the modified vector $\widehat{\wbold}_m$ is added to the set of $k_{nn}$ vectors, and the selected random item is removed from the buffer $b_1$ in lines 22 and 23 respectively. If all elements of $b_1$ are processed and the $k_{nn}$ buffer is not still filled with $k$ vectors, the quantization level variable $L$ is increased to 2 and the procedure continues by initially filling the buffers $b_1$ and $b_2$ and then sequentially adding random nearest vectors $\widehat{\wbold}_m$ each having 2 quantization level differences with the original vector $\widehat{\wbold}_m$. The algorithm continues until the $k_{nn}$ is filled with $k$ nearest neighbors of $\wbold_m$. One could verify that the complexity of the presented knn algorithm is $\mathcal{O}(kN')$.
	
	\begin{algorithm}[t]
		\caption{\small\!: Proposed low-complexity knn algorithm}
		\begin{algorithmic}[1]
			\State \textbf{Initialize:} Let $k_{nn}=\{\}$ be the set of knn vectors, $b=\{\}$ be a tensor of temp buffers, $k'=0$ be the current neighbor index, and $L=0$  indicating the level of  quantization difference between the neighbor vector $\widehat{\wbold}_m$ and the original vector $\wbold_m$;
		\While{$k'<k$} 
        \If{$b$ is empty}
        
            \State $L\leftarrow L+1$;
            
            \For{$l=1$ to $L$}
            
                \State $b_l\leftarrow[1,2,...,N',-1,-2,...,-N']$; 
            
            \EndFor
        \EndIf
        \State
        Initialize the candidate neighbor $\widehat{\wbold}_m \leftarrow \wbold_m$;
        \For{$l=1:L$}
            
            \State
            Set $r_l$ a random integer between 1 and $|b_l|$;

            \State
            $i\leftarrow|b_l(r_l)|, s\leftarrow sign(b_l(r_l))$;

            \State
            Set $ind_l$ as the index of the element $\widehat{\wbold}_{mi}$ in $\mathcal{W}$.

            \If{$s>0$ and $ind_l<2^r$}
            
            \State
            $\widehat{\wbold}_{mi}=\mathcal{W}(ind_l+1)$;

            \ElsIf{$s<0$ and $ind_l>1$}

            \State
            $\widehat{\wbold}_{mi}=\mathcal{W}(ind_l-1)$;
            
            \EndIf

        \EndFor
        \If{$\widehat{\wbold}_m\neq \wbold_m$}
            \State 
            $k'=k'+1$;
            \State
            Add $\widehat{\wbold}_m$ to the $k_{nn}$ buffer.
            \State
            $b_l=b_l/ \{b_l(r_l)\}$ for $1\leq l\leq L$
        \EndIf

        \EndWhile
        \end{algorithmic}
	\end{algorithm}

\subsubsection{Critic}
As seen in Fig. \ref{fig:structure_drl}, the TD3 structure is composed of two NNs with parameters $\boldsymbol{\theta}^{Q_i},i\in\{1,2\}$, each estimating the corresponding $Q$ value. The $Q$ value is a meter used in Deep $Q$ Learning schemes which estimates how good an action is. In the standard TD3 network, the state and action are first concatenated and serve as the input of both critic NNs and the corresponding Q value is then obtained in the output of the networks, i.e.,
\begin{align}
    \label{eq:critic_main}
    Q_i(., .|\boldsymbol{\theta}^{^{Q_i}}): \mathcal{S} \times \mathcal{A} \rightarrow \mathbb{R}, \ \forall i\in\{1,2\}
\end{align}
Considering the state and action space in our problem, \eqref{eq:critic_main} corresponds to \eqref{eq:critic_main2} as follows:
\begin{align}
    \label{eq:critic_main2}
    Q_i(.,.|\thetabold^{^{Q_i}}): \mathcal{W}^{N'} \times \mathcal{W}^{N'}\rightarrow \mathbb{R}, \ \forall i\in\{1,2\}
\end{align}
Finally the overall $Q$ value is obtained as $Q=\min{(Q_1,Q_2)}$.

\subsubsection{Agent} The agent is responsible for training the actor and critic NNs and controlling the trend of the learning process and convergence of the training scheme. In TD3, in addition to the actor and critic NNs, there exists a \textit{target actor} NN denoted by $\pibold^{a,t}(.|\thetabold^{a,t})$ and two \textit{target critic} NNs denoted by $Q_i^t(.,.|\thetabold^{Q_i,t}), \forall i\in\{1,2\}$. At each time step $n$, using the current observation state $\sbold^{(n)}$, the current action is selected as $\abold^{(n)}=\pibold^a(\sbold^{(n)}|\thetabold^a)+\nu^a$ where $\nu^a$ is a stochastic exploration noise which is obtained based on the noise model.
In the standard TD3 problem, the obtained action $\abold^{(n)}$ is  applied to the environment, and then the reward $r^{(n)}$ is calculated and the next state $\sbold^{(n+1)}$ is observed. The experience $\left(\sbold^{(n)},\abold^{(n)},r^{(n)},\sbold^{(n+1)}\right)$ is then stored and added to the experience buffer. In our problem, however, the action $\abold^{(n)}$ is additionally quantized and the knn actions are computed. After that, the best action from the set of knn actions (the one corresponding to the highest $Q$ value) is replaced by the original action $\abold^{(n)}$ before storing it in the experience buffer. A random minibatch $\left(\sbold^{(k)},\abold^{(k)},r^{(k)},\sbold^{(k+1)}\right)$ of size $K$ is then sampled from the experience buffer and for each sample $k$ of the minibatch, the target $y^{(k)}$ is calculated based on Q-learning principle as follows:
\begin{multline}
    \label{eq:yk}
    y^{(k)}=r^{(k)}
    \\
    +
    \gamma \min_{i\in \{1,2\}} \left( Q_i^t(\sbold^{(k+1)},
    f^{clip}(\pibold^{a,t}(\sbold^{(k+1)}|\thetabold^{a,t})+\nu^a)|\thetabold^{Q_i,t})
    \right)
\end{multline}
where $f^{clip}$ is a clipping function that limits the computed action plus noise to the minimum and maximum allowed thresholds. Once $y^{(k)}$ is computed for all $K$ minibatch samples of the experience buffer, it is time to update the parameters of the actor and critic NNs. The parameters of the critic NNs are updated at each time step $n$ by minimizing the following mean square error loss function:
\begin{align}
    \label{eq:L}
    L_i^{(n)}=\frac{1}{K} \sum_{k}\left( 
        y^{(k)}-Q_i(\sbold^{(k)},\abold^{(k)}|\thetabold^{Q_i})
    \right)^2, \ \forall i\in\{1,2\}
\end{align}
For each $T_1$ step, the actor parameters are updated using the sampled policy gradient to maximize the expected discounted reward as follows:
\begin{multline}
    \label{eq:GJ}
    \nabla_{\thetabold^a} J(\thetabold^a)) \approx \frac{1}{K} \sum_{k} 
    \nabla_{\abold}
    \left(
    \min_{i\in\{1,2\}}{Q_i(\sbold^{(k)},\abold|\thetabold^{Q_i})}
    \right)
    \\
    \times
    \nabla_{\thetabold^a} \pibold^a(\sbold^{(k)}|\thetabold^a)
\end{multline}
and finally, in every $T_2$ step (where $T_2>T_1$), the target actor and target critic NNs' parameters are updated as follows:
\begin{subequations}
\label{eq:target_update}
\begin{align}
    \thetabold^{a,t}&\leftarrow\tau \thetabold^a + (1-\tau)\thetabold^{a,t}  
    \\
    \thetabold^{Q_i^t}&\leftarrow\tau \thetabold^{Q_i} + (1-\tau)\thetabold^{Q_i^t}, \ \forall i\in\{1,2\}  
\end{align}
\end{subequations}
where $\tau$ is the target smoothing factor which is a small positive value. A detailed representation of the stated TD3 scheme is expressed in Algorithm 2.
\begin{algorithm}[t]
		\caption{\small\!: TD3-DRL scheme for each sub-array module $m$}
		\begin{algorithmic}[1]
           
		  %\State \textbf{Initialization:}

            \State
            Initialize actor and critic NNs' parameters $\thetabold^{a}, \thetabold^{Q_1}, \thetabold^{Q_2}$ with random weights and set initial values of target NNs' parameters with $\thetabold^{a,t}\leftarrow\thetabold^{a}, \thetabold^{Q_1^t}\leftarrow\thetabold^{Q_1}, \thetabold^{Q_2^t}\leftarrow \thetabold^{Q_2}$.

            \State
            Initialize the exploration noise random process $\nu^a$, assign required memory to experience buffer $\mathcal{K}$,  and set proper values to minibatch size $K$, target smoothing factor $\tau$, discount factor $\gamma$, and update frequencies $T_1$ and $T_2$.

            \State
            Set an initial random state as $\sbold^{(1)}=\wbold^{(1)}_m\in\mathcal{W}^{N'}$.

            \For{\textbf{each} time-step $n$}
                
                \State  Set $\widehat{\abold}=\pibold^a(\sbold^{(n)}|\thetabold^a)+\nu^a$.

                \State
                Quantize the action $\widehat{\boldsymbol{a}}$, obtain knn actions from Algorithm 1, and obtain the final action $\abold^{(n)}$ as the one corresponding to the highest $Q$ value.

                \State Apply $\abold^{(n)}\equiv \wbold^{(n)}$ to the environment, get the reward $r^{(n)}$ from \eqref{eq:reward}, and the next state as $\sbold^{(n+1)}=\abold^{(n)}$.

                \State
                Add $\left(\sbold^{(n)},\abold^{(n)},r^{(n)},\sbold^{(n+1)}\right)\equiv \left(\wbold^{(n-1)},\wbold^{(n)},r^{(n)}\right)$ to the experience buffer $\mathcal{K}$.  

                \State
                Select a random minibatch $\left(\sbold^{(k)},\abold^{(k)},r^{(k)},\sbold^{(k+1)}\right)$ from $\mathcal{K}$ of size $K$.

                \State
                For each $k$ in the minibatch, update  $y^{(k)}$ from \eqref{eq:yk}.
                
                \State
                Update critic  parameters $\thetabold^{Q_i}, \forall i\in\{1,2\}$ using \eqref{eq:L}.

                \State
                Every $T_1$ step, update actor parameters $\thetabold^a$  using \eqref{eq:GJ}.

                \State
                Every $T_2$ steps update target NNs' parameters $\thetabold^{a,t}$ and $\thetabold^{Q_i^t}, \forall i\in\{1,2\}$ from \eqref{eq:target_update}.
                
            \EndFor
           
    \end{algorithmic}
\end{algorithm}

\subsection{Calculation of the overall beamforming vector from sub-array modules' beamforming vectors}
\label{sec:overall_codebook}
Once the optimal vector $\wboldopt_m$ for each module $m$ is obtained through the DRL scheme in Algorithm 2 (i.e., after all DRL networks converge), the optimal overall beamforming vector $\wboldopt$ should be calculated. It is seen from Fig. \ref{fig:structure_hardware} that if the vectors $\wboldopt_1$, ..., $\wboldopt_M$ are obtained in a way that all modules' output signals $x_1$, $x_2$, ..., $x_M$ are of the same phase, the resulting summation signal $x$ will be maximized, and thus the optimal beamforming vector is simply obtained through the concatenation of all modules' beamforming vectors, i.e., $\wboldopt=[\wboldopt_1, \wboldopt_2, ..., \wboldopt_M]$. In practice however, after finding the vectors of $\wboldopt_m$ for each module $m$ through the corresponding DRL network, there might exist an offset between the phases of $x_m$ and $x_n$ for each module $m\neq n$. 
 For the case of analog beamforming vector space (i.e., when $\measuredangle \wbold \in [0,2\pi]^{N'}$ and no quantization is required), one could verify that the following offset phase shift for all elements of $\wboldopt_m$ (for any $m\neq 1$) aligns the phases of $x_1$ and $x_m$ leading in the maximum power of the summation of signals.
 \begin{align}
     w^{'\mathrm{opt}}_{mi}=w^{\mathrm{opt}}_{mi} \times \exp\left({j}(\measuredangle x_1 - \measuredangle x_m)\right), \ \ \forall i,m\neq 1,
 \end{align}
 Finally, the following beamforming vector will align all phases leading to the maximum measured power.
 \begin{align}
    \label{eq:woptconcat}
    \wboldopt=[\wboldprimeopt_1, \wboldprimeopt_2, ..., \wboldprimeopt_M]     
 \end{align}
 For the case of quantized action-space, the elements of $\wbold_m^{'\mathrm{opt}}$ can be obtained as follows:
 \begin{multline}
\label{eq:wprimopt_quantized}w^{'\mathrm{opt}}_{mi}=w^{\mathrm{opt}}_{mi} \times \exp
     \left({j}
     \measuredangle \argmin_{\phi\in\phi^{\mathrm{quan}}}{\left| \phi- (\measuredangle x_1 - \measuredangle x_m)\right|}
     \right),
     \\
     \forall i,m\neq 1.
 \end{multline}
 and $\wboldopt$ is obtained from \eqref{eq:woptconcat}.
Based on what is stated so far, the proposed low-complexity modular DRL SBF algorithm is presented in Algorithm 3.
{\color{black}
\begin{remark}
    The proposed SBF mechanism is highly scalable in terms of the proposed hardware structure and software complexity. In terms of hardware structure, providing a lower BFR and/or a focal point at higher distances requires an aperture of higher dimension, which necessitates a higher number of subarray modules, without any change in the hardware structure. From the software and training point of view, each subarray module has its training agent which is trained independently of others and does not impose any processing load on other modules. This way, for example, if we consider that each module is equipped with an independent processing unit, the number of required processing units is linearly increased with the number of subarray modules/array elements, and thus, the required processing resources are a linear function of the number of array elements, which is completely scalable. The final block for calculating the overall beamfocusing vector from subarrays' vectors is a single-shot light computation component which does not affect the overall complexity.
\end{remark}
}

\begin{algorithm}[t]
    \caption{\small\!: Distributed DRL-based SBF Algorithm}
    \begin{algorithmic}[1]
        \State
        Initialize DRL network parameters for each antenna module $m\in\M$ according to steps 1-3 of Algorithm 2.

        \For{each time step $n$ until all DRLs converge}
         
        \For{each antenna sub-array $m\in\mathcal{M}$}
        
      \State 
        Train the DRL-NN parameters of module $m$ and obtain $\wbold^{(n)}_m$ from steps 5-13 of Algorithm 2.
        \EndFor
        
    \EndFor

    \State
    Set $\wboldopt_m\leftarrow \wbold_m^{(n)}, \ \forall m\in\M$.

    \State 
    Calculate $\wboldprimeopt_m$ from \eqref{eq:wprimopt_quantized}, $\ \forall m\in\M$.

    \State
    Obtain the final beamforming vector $\wboldopt$ from \eqref{eq:woptconcat}.
    
    \end{algorithmic}
\end{algorithm}
%\subsection{Extension to multi-focal point scenario}
%Extending the 
 \section{Numerical Results}
% In what follows, we evaluate the performance of the proposed SBF structure. Unless otherwise stated, the default simulation parameters are given in table *** 
\begin{table*}
		\centering
		\caption{Simulation Parameters}
		\begin{tabular}{|l|l|l|l|}
	\hline
		\textbf{Parameter}
        & 
        \textbf{Description}
        &
        \textbf{Parameter}  
        &
        \textbf{Description} 
        \\
        \hline
        Frequency
        &
        $28$ GHz
        &
        Path-loss exponent ($\alpha$)
        &
        $2.7$
        \\
        Reflection coefficient ($\beta_{nl}, \forall n,l$) 
        & 
        $0.1$
        &
        PM sub-array elements	  &
        $6\times 6$
        \\
        ELPM number of modules
	&
        100
        &
        Room Dimensions
        & 
        $4 \times 4 \times 3\ \mathrm{m}^3$                  
		\\
        Exploration noise variance ($\nu^a$)
        & 
        $0.5$
        &
        Exploration noise decay rate
        &
        $10^{-5}$
        \\
        Exploration noise minimum
        &
        $10^{-3}$
        &
        Target policy variance
        &
        $0.1$
		\\	
        Target policy decay rate
        & 
        $10^{-4}$
        &
        $(T_1,T_2)$
        &
        (1,3)
        \\		
        \hline
        \end{tabular}
        \label{tbl:simulation_params}
\end{table*}
 \begin{figure}
		\centering
		\includegraphics [width=244pt]{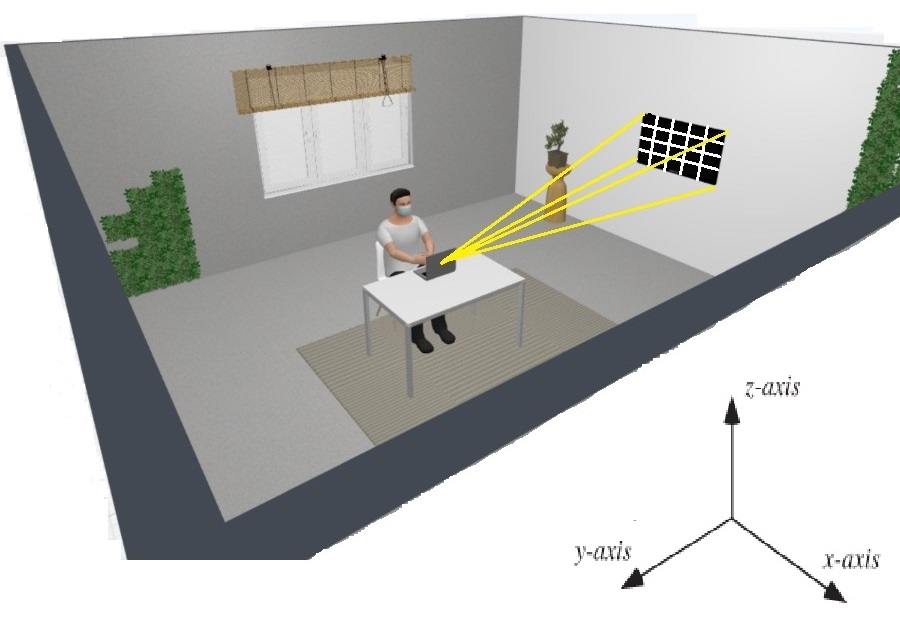} \\
		\vspace{-15pt}
		\caption{Simulation environment} 
		\label{fig:room}
% 		\vspace{-10mm}
\end{figure}
To verify the performance of the proposed structure, we have considered a large-scale 3600-element ELPM consisting of $M=10\times10=100$ sub-arrays, each having $N=6\times 6=36$ elements, wherein the horizontal and vertical distance between adjacent antenna elements is $\lambda/2$.
The reason for choosing sub-arrays of $6\times 6$ elements is that they can be trained in a reasonable number of iterations approaching very close to the global optimal target, as will be shown in the simulation results. Unless otherwise directly stated, the default simulation parameters are listed in table \ref{tbl:simulation_params}.
As seen in Fig. \ref{fig:room}, a $4\times4\times 3 \ \mathrm{m}^3$  room with walls, ceiling, and roof reflection coefficient of 0.1 is considered wherein the UE is located in front of the ELPM at distance $y$ and height $z=1.4$ m, and the ELPM first element (bottom-left) is located at the point $(1,0,1.5)$ m. Different values for the distance $y$ will be considered in the simulations.
For each sub-array module, {\color{black}inspired from \cite{li2021deep,pingpong}}, we have considered a DRL in which the actor NN starts with a normalization input layer followed by a $16N$-neuron fully connected layer (FCL), a Rectified Linear Unit (RELU) layer, a $16N$-neuron FCL, a RELU layer, and then
an FCL with $N$-neuron, a hyperbolic tangent (tanh) layer, and a scaling output layer to map the output space domain into $[-\pi, \pi]^N$. The critic NN for each of the two agents of the TD3-DRL network starts with $2N$-neuron input layer concatenating the observation and action inputs, followed by $32N$-neuron FCL, RELU, $16N$-neuron FCL, a tanh layer, and finally a 1 neuron fully connected output layer.

\subsection{Performance of the proposed scheme for a single sub-array PM}
First, we examine the performance and convergence of the proposed scheme using TD3 and its outperformance over DDPG
%\footnote{The DDPG structure for the simulations has been considered the same as TD3 except that instead of the two NNs of the TD3 critic network, only one is considered for DDPG.}
for a single sub-array aperture.
 \begin{figure}
    \centering
    \begin{tabular}{c} \includegraphics[width=254pt]{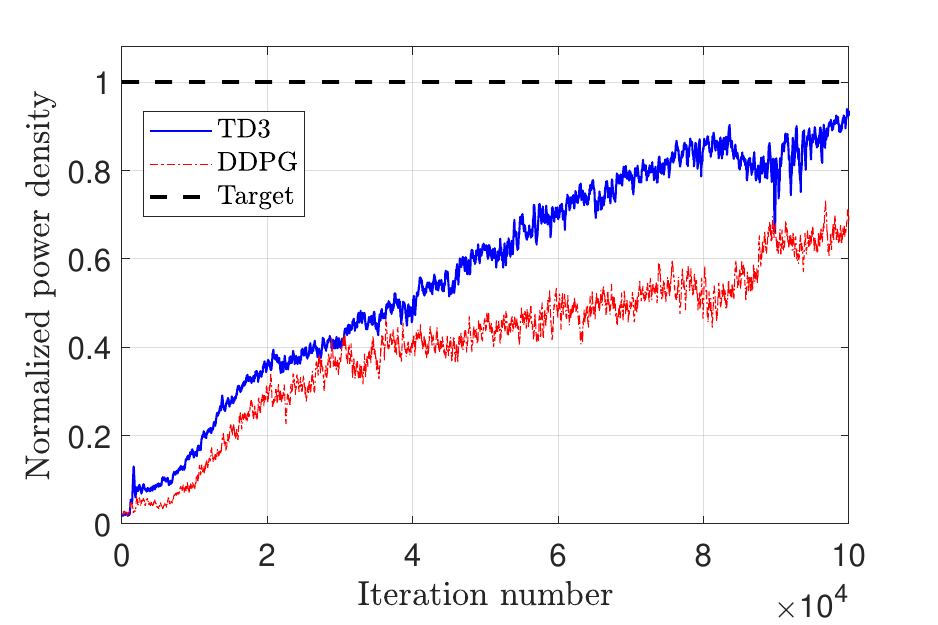}
    \\%[\abovecaptionskip]
    (a): $y=1.4$m
    \\
    \includegraphics[width=254pt]       {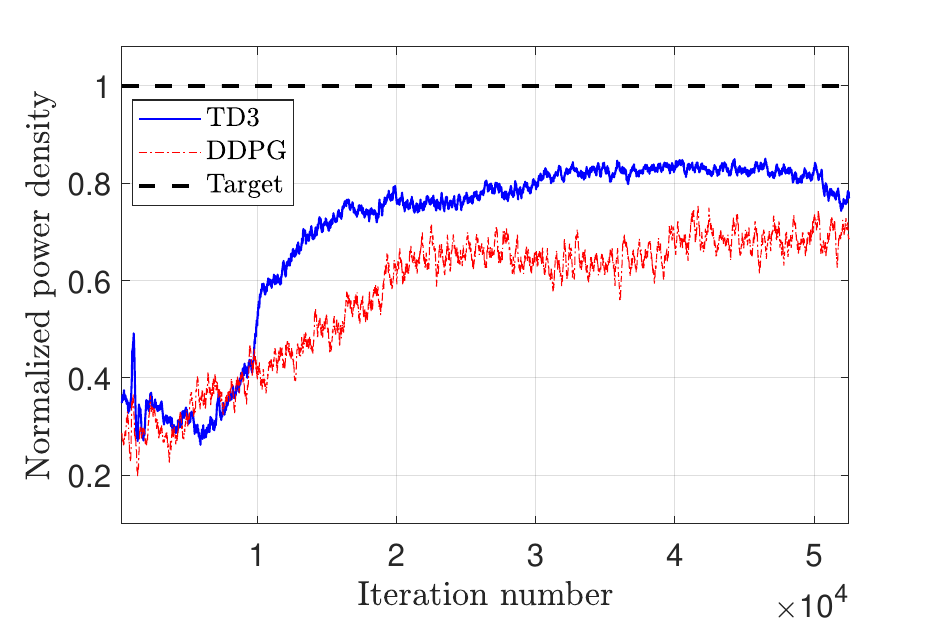}
    \\%[\abovecaptionskip]
    (b): $y=4$m
    \end{tabular}
    \caption{Normalized power density per training iteration number for TD3 and DDPG DRL schemes for the ELPM sub-array corresponding to the first row and first column of the PM array  for the distances of $y=1.4$ m and $y=4$ m.} 

    \label{fig:1}
    % \vspace{-10mm}
\end{figure}
\begin{figure}
  \centering
    \includegraphics[width=254pt]{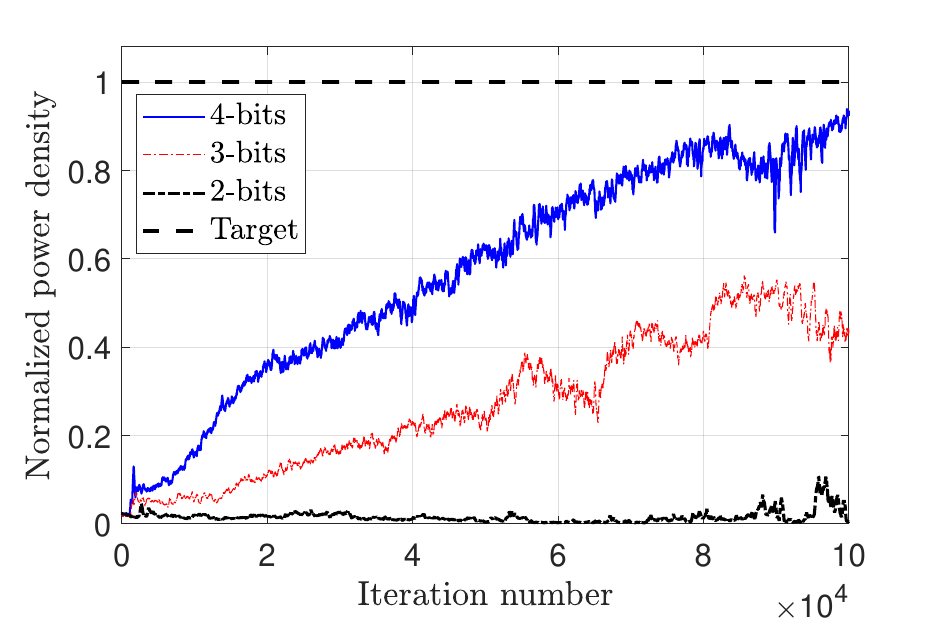}
    \caption{Normalized power density per training iteration number for the proposed TD3 schemes for the PM sub-array corresponding to the first row and first column of the ELPM array modules considering $2$, $3$, and $4$ bit phase shifters  ((sub-array corresponding to the first row and first column of the ELPM array modules)).} 
    \label{fig:sim_per_noof_bits}
\end{figure}
The received normalized power density at the UE location per training iteration number for TD3 and DDPG schemes is plotted in Fig. \ref{fig:1} for a single sub-array and two different UE distances of $1$ m and $4$ m from the ELPM aperture. 
First, it is seen that TD3 outperforms DDPG in terms of both convergence speed and received power level due to the multiple critic NNs and target policy smoothing.
% The use of multiple critics in TD3 helps to reduce the overestimation of the Q-value and leads to more stable training. In terms of convergence speed, TD3 has been shown to converge faster than DDPG on a variety of tasks. This is largely due to the use of which helps to regularize the policy during training and prevents it from overfitting to the Q-values. Additionally, TD3 uses delayed updates for the critic network, which further improves stability and convergence speed. As can be seen in Figures 4a and 4b, the learning curves in TD3 and DDPG methods are close to each other. Then, at a separating boundary, they are far away from each other. This separating boundary is the number of iterations required for training policies. After training the policies, the TD3 algorithm shows its advantages.
For the case of 1.4 m UE distance, after about $40$k iterations, the TD3 algorithm overtakes DDPG and reaches 90 percent of the target in about $100$k repetitions. Meanwhile, DDPG has reached about 70 percent of the target. By increasing the distance to 4 meters (and getting closer to the far-field scenario), the same scenario as before is repeated, with the difference that the policy training speed has been increased. In general, all sub-array modules converge with a scenario close to the above. Due to the superiority of the TD3 algorithm, from now on, we only consider TD3 in the next simulation parts.

 To demonstrate how increasing the resolution of quantized phase shifters influences the performance and convergence speed of the proposed scheme, we have compared the simulation results for 2-bit, 3-bit, and 4-bit phase shifters as illustrated in Fig. \ref{fig:sim_per_noof_bits}. It is seen how increasing the quantization resolution improves the focused power level measured at the UE location. %This is because signal convergence at the UE location is highly dependent on the phase of the elements. If the phase resolution decreases, the quantization error increases, and the power at the SPOT location decreases. However, the quantization error can slow down the algorithm's convergence significantly. In the case of two-bit phase mode, the algorithm may experience difficulty in converging, requiring a high number of repetitions.
For example, after about $85$k iterations, the employed  4-bit and 3-bit phase shifters lead respectively to performance measures of about 90 percent and 52 percent of that relating to the ideal target power, however, the 2-bit phase shifter is seen to be almost of no use.

\begin{figure}
  \centering
    \includegraphics[width=254pt]{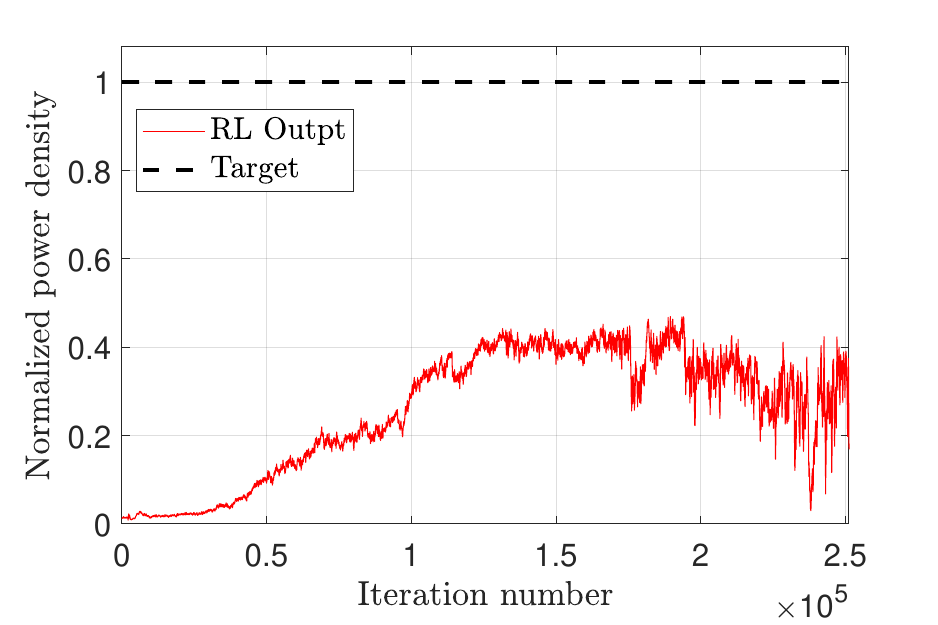}
\caption{Training trend for a $8\times 8$ PM} 
\label{fig:sim_8_8}
\end{figure}
\subsection{SBF performance of the proposed scheme for ELPM}
In the previous subsection, we analyzed the performance of the DRL-based beamfocusing for a single $6\times 6$ sub-array PM. A 3D spot beamfocusing requires an ELPM consisting of a much larger number of antenna elements (e.g., 3600 elements as considered in our simulations). The question that arises is: what is the cost of not using the proposed distributed method? Is it possible to apply the proposed single DRL-based sub-array structure to a large-scale array? As will be shown here, increasing the number of antenna elements leads to unaffordable computational complexity and convergence issues. To verify this, let's consider a single $4$-bit sub-array consisting of only $8\times 8$ elements instead of the original $6\times 6$ one. The performance per learning epoch is illustrated in Fig. \ref{fig:sim_8_8}. It is seen here that the highest performance is obtained after a large number of iterations (around $200$k) which is only about $47$ percent of the target value; remember that for the case of $6\times 6$ sub-array, we reached a performance of about $90$ percent of the target value after about $80$k iterations. This reveals that the achievement of SBF (which requires a few thousand antenna elements as shown in the following) is not possible through a single module, rather, it requires a distributed learning structure.

Now we study the exploitation of a large-scale $3600$ element ELPM consisting of $10\times 10$ sub-arrays each having $6\times 6$ elements, wherein the DRL-based SBF is implemented through Algorithm 3.
%By ensuring the efficiency of the proposed TD3 beamfocusing scheme applied to the sub-arrays through Algorithm 2, the proposed distributed DRL-based SBF method for the aforementioned ELPM (consisting of 100 subarrays and 3600 antenna elements)  is implemented according to Algorithm 3, and its performance is analyzed in the following.
%, wherein, each of the 100 sub-arrays searches for the optimal phase distribution for their elements, in parallel with the other sub-arrays, by applying the TD3 algorithm. Then, based on Algorithm 3, the optimal state of each sub-array in each time step results in the optimal PM state.
\begin{figure*}[t!]
  \begin{tabular}{c c c}
  \hspace{-10pt}
    \includegraphics[width=0.33\linewidth]{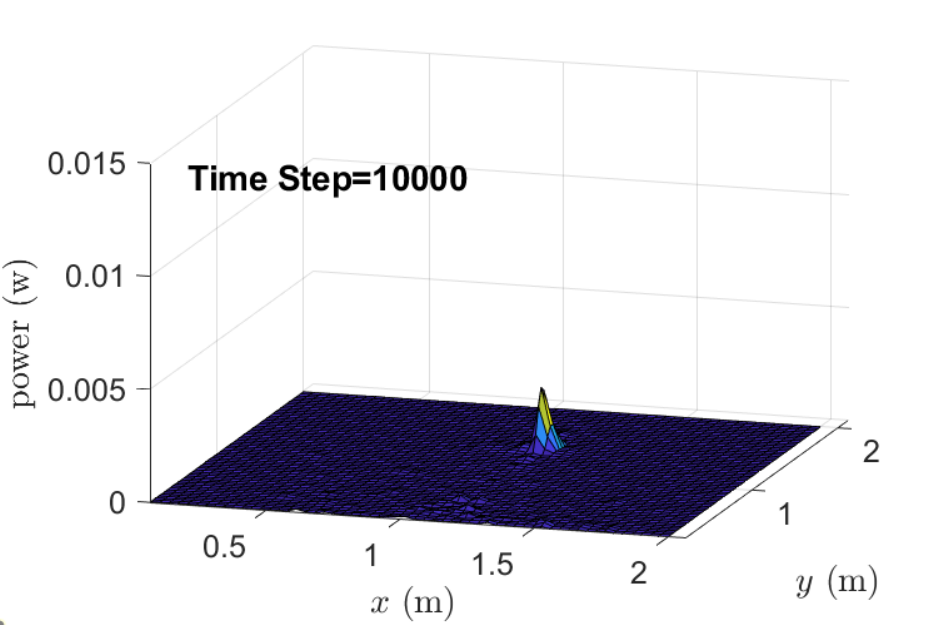}
    &
     \hspace{-10pt}
    \includegraphics[width=0.33\linewidth]{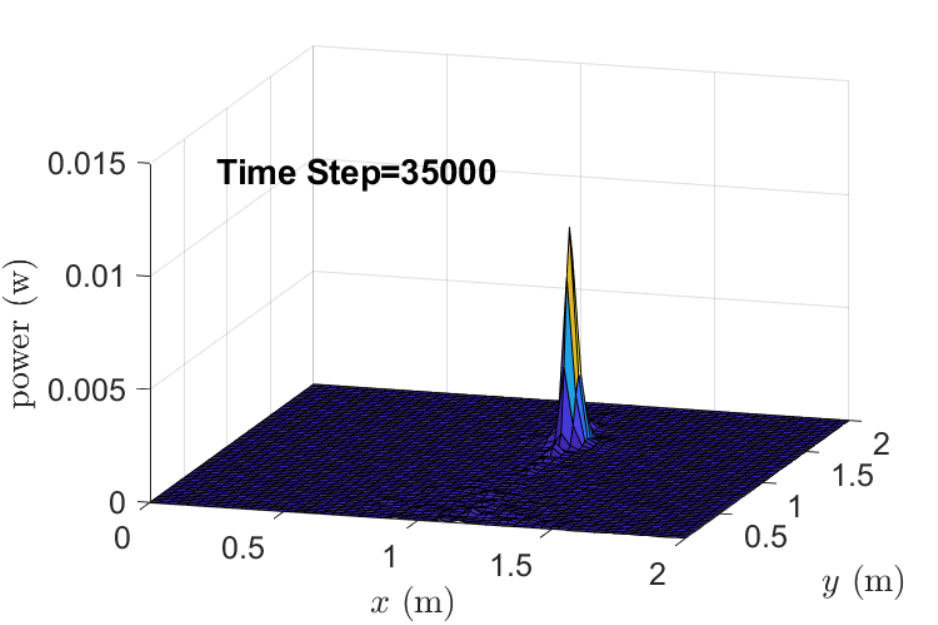}
    &
     \hspace{-10pt}
    \includegraphics[width=0.33\linewidth]{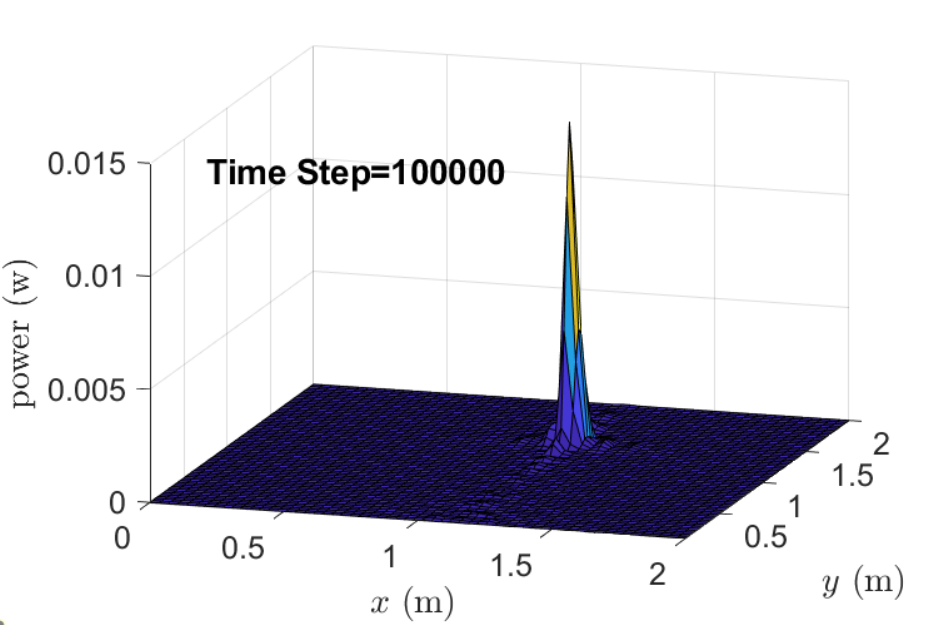}
\\
    (a) & (b) & (c) 
\end{tabular}
  \caption{Spatial power pattern around the UE location for the action vectors of the proposed DRL for the ELPM array corresponding to different training steps (a): iteration number: 10k, (b): iteration number: 35k, and (c): iteration number: 100k.}
  \label{fig:sim_3D_per_iteration}
\end{figure*}
First, we have demonstrated the spatial power distribution resulting from the radiating ELPM in the room environment in the planar surface parallel to $xy$-plane around the UE location point (at $z=1.4$ m)  for three iteration snapshots of $10$k, $35$k and $100$k as depicted in Fig. \ref{fig:sim_3D_per_iteration}.
It is seen that as the learning proceeds forward, the spatial power distribution moves toward the \textit{spot} point beamfocusing behavior by observing that the higher the training epoch, the more sharp spot-like beam at the DFP.
\begin{figure}
    \centering
    \includegraphics[width=254pt]
    {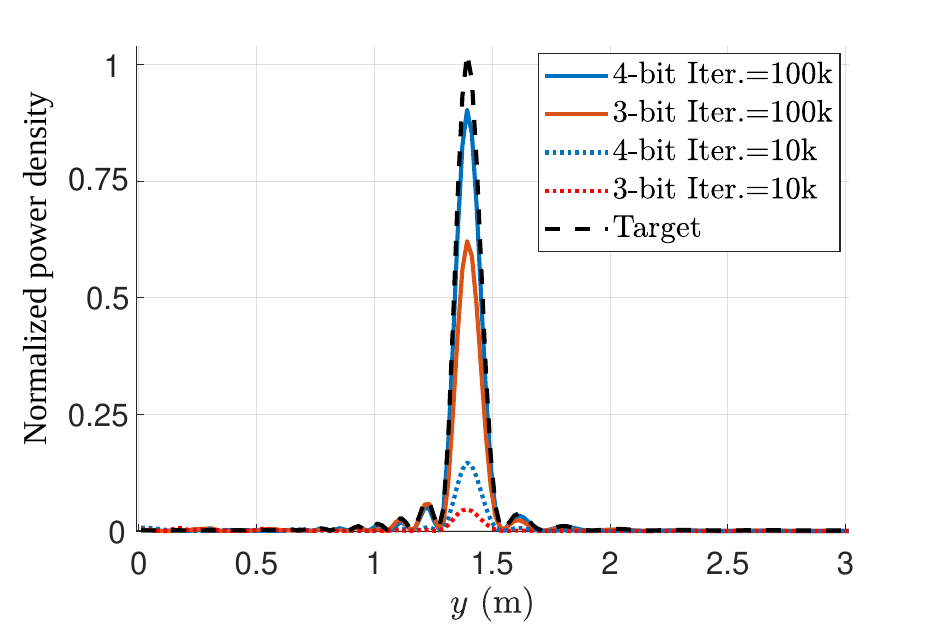}
    \caption{Power density per distance ($y$) for 3-bit and 4-bit phase shifters at iterations of 10k and 100k} 
    \label{fig:sim_2D_per_distance}
\end{figure}
This power concentration is more clearly illustrated in the 2D diagram of Fig. \ref{fig:sim_2D_per_distance}, wherein the power per distance ($y$) is depicted for 3-bit and 4-bit quantizers at the epochs of $10$k and $100$k. It is seen here how increasing the quantization level and training iteration leads to a sharper spot-like beam concentration. For example, it is seen that after $100$k training epochs for the 3-bit and 4-bit phase-shifters, the ratio of the peak measured power to the target power is about 60 and 90 percent, and the BFR corresponding to $\eta=80$\% at the UE location is about $8$ cm and $7$ cm respectively. 

\begin{figure}
  \centering
\includegraphics[width=254pt]{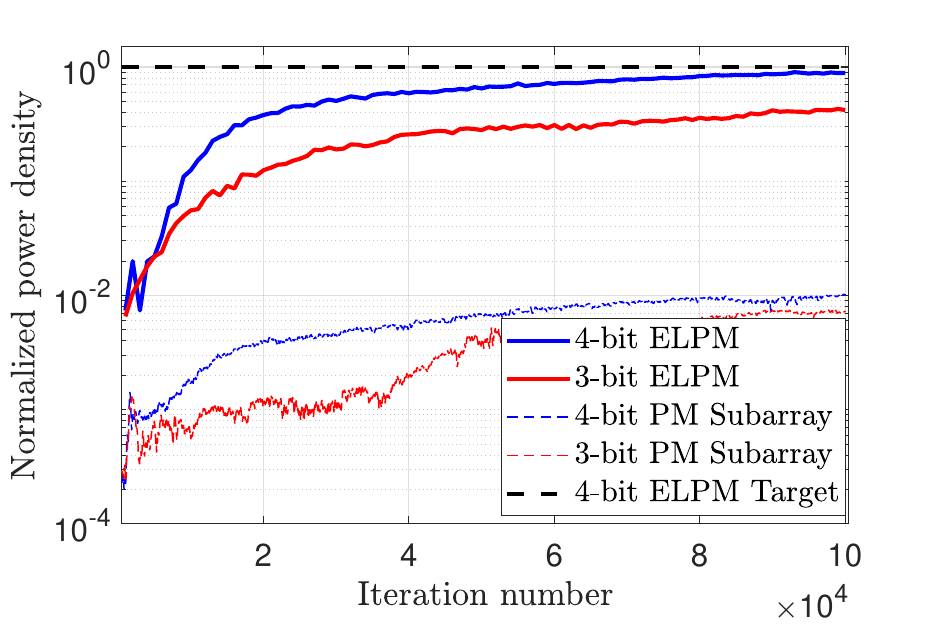}
\caption{Performance comparison of the proposed distributed modular scheme for large-scale $10\times 10$ module ELPM array versus that of a single module PM using $3$ and $4$ bit phase shifters.} 
\label{fig:sim_array_vs_subarray}
\end{figure}
To numerically verify how the proposed distributed modular structure leverages the SBF performance in comparison to the deployment of a single-module DRL-based 3D beamforming, we have compared the performance results of the 100 sub-array ELPM to that of a random single sub-array as seen in Fig. \ref{fig:sim_array_vs_subarray}. For each of the 3-bit or 4-bit phase shifters, the implementation of 100 sub-arrays leverages the measured power at the focal point of the order of 100 times the one obtained from a single sub-array, which reveals the effectiveness of the proposed structure.

\begin{figure}
  \centering
\includegraphics[width=254pt]{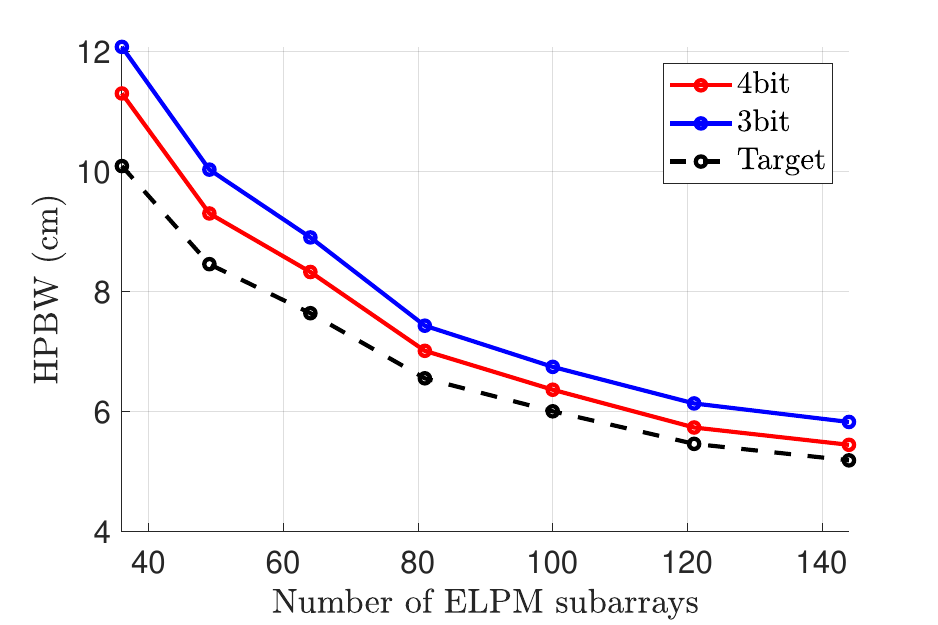}
\caption{Performance comparison of the proposed distributed modular scheme versus the number of subarray modules.} 
\label{fig:sim_hpbw}
\end{figure}
{\color{black}
Finally, we elaborate on the impact of the number of subarray modules on the system’s performance, as seen in Fig. \ref{fig:sim_hpbw}. The figure shows the achieved half-power beamwidth (HPBW) as a simple criterion showing how small the size of the beam is around the focal point, versus different numbers of subarrays ranging from $6\times 6=36$ to $12\times 12=144$. It is observed how increasing the number of subarrays results in a more concentrated power of the signal around the focal point in a non-linear way. For instance, increasing the number of subarrays 4 times from 36 to 144 results in the HPBW being approximately halved for each of the 3-bit TD3-DRL, 4-bit TD3-DRL, and target scenarios.
}
\section{Conclusion}	
\label{sec:conclusions}

In this paper, we presented a novel CSI-independent DRL-based scheme for Fresnel zone spot beamfocusing using ELPMs. To overcome the complexity and convergence issues, we introduced a distributed structure consisting of several sub-array PMs each equipped with a TD3-DRL optimizer. This modular setup enabled collaborative optimization of the transmitted power at the DFP, significantly reducing computational complexity while enhancing the convergence rate. The proposed scheme makes CSI-independent SBF for ELPMs quite feasible and computationally affordable in practice; this has not been the case in existing works in the literature so far. We provided multiple
numerical results to demonstrate the performance and benefits of our proposed structure.
%{\color{black}
%Several intriguing research areas beckon for future exploration. These include the application of transfer learning for SBF, integrating the CSI of a set of sub-channels with the ML-based SBF solution (assuming that the CSI of a subset of ELPM channels is available), investigation of ML-based SBF with non-exact CSI rather than the pure CSI-independent scenario considered in this work, and finally the extension of the work for realizing ML-based SBF with multiple focal points.
%}
Several intriguing research areas beckon for future exploration. These are as follows:
\begin{itemize}
    \item In this work, we considered that each subarray learns the optimal beamfocusing subvector with no initially trained policy. In practice, the beamfocusing submatrices of different subarrays are correlated to each other. An interesting research direction is to leverage this correlation by employing various transfer learning techniques.

    \item A completely CSI-independent algorithm inherently possesses a degree of complexity. When the exact or inexact CSI for a subset of array elements is available, the proposed DRL-based scheme can be modified to incorporate such data into the learning process to reduce computational complexity and enhance speed and performance.

    \item Other research directions include exploring how the whole array be divided optimally into the set of subarrays (i.e., finding optimal values of $N$ and $M$), investigating non-uniform distribution of each subarray antenna elements throughout the whole aperture, and finally considering scalable multi-focal SBF through multi-agent DRL-based schemes.  
\end{itemize}

%\begin{figure}
%  \centering
%\includegraphics[width=254pt]
%{pictures/simulation/Fig9/fig9_HalfPower_Beamwidth_lin_subarray.pdf}
%\caption{Halfpower Beamwidth.} 
%\label{fig:2}
%\end{figure}

%\begin{figure}
%  \centering
%\includegraphics[width=254pt]
%{pictures/simulation/Fig10/fig10_SBF_HPBW.pdf}
%\caption{Halfpower Beamwidth.} 
%\label{fig:2}
%\end{figure}

% 	\footnotesize
\bibliographystyle{IEEEtran}
\bibliography{MyBib}

% \balance{
 \vspace{-10pt}
\begin{biography}[{\includegraphics[width=1in,height
=1.25in,clip,keepaspectratio]{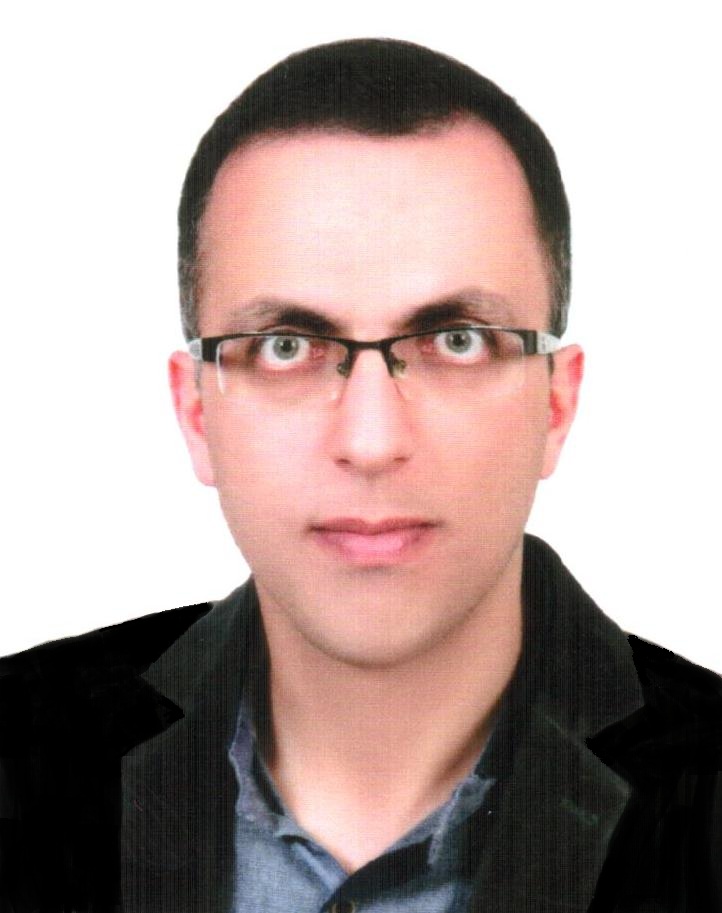}}]{Mehdi Monemi} (Member, IEEE)
		received the B.Sc., M.Sc., and Ph.D. degrees all in electrical and computer engineering from Shiraz University, Shiraz, Iran, and Tarbiat Modares University, Tehran, Iran, and Shiraz University, Shiraz, Iran in 2001, 2003 and 2014 respectively. After receiving his Ph.D., he worked as a project manager in several companies and was an assistant professor in the Department of Electrical Engineering, Salman Farsi University of Kazerun, Kazerun, Iran, from 2017 to May 2023. He was a visiting researcher in the Department of Electrical and Computer Engineering, York University, Toronto, Canada from June 2019 to September 2019. He is currently a Postdoc researcher with the Centre
for Wireless Communications (CWC), University of Oulu, Finland. His current research interests include resource allocation in 5G/6G networks, as well as the employment of machine learning algorithms in wireless networks.
	\end{biography}
 \vspace{-10pt}
\begin{biography}[{\includegraphics[width=1in,height
=1.25in,clip,keepaspectratio]{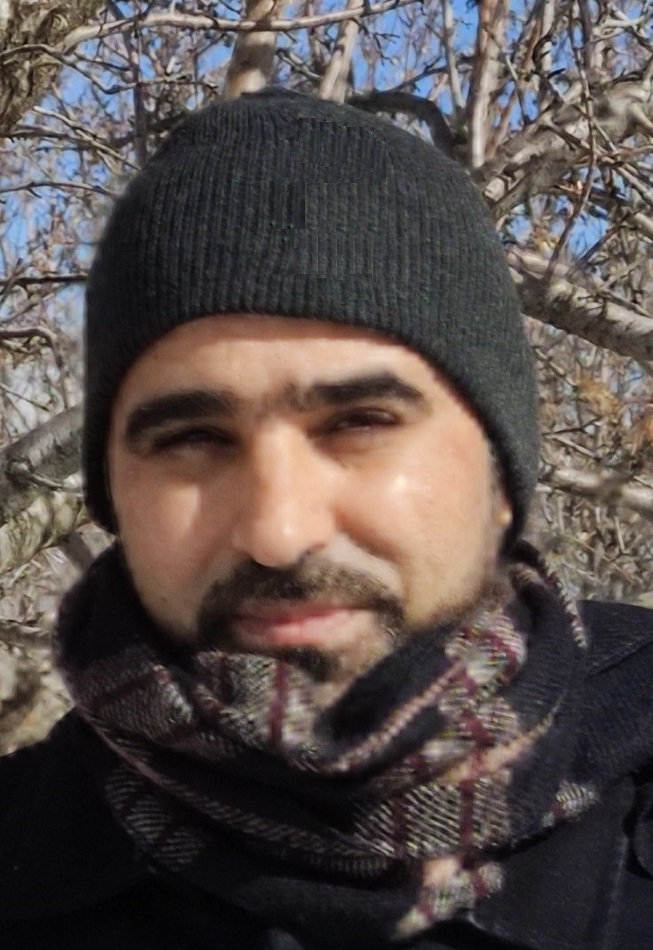}}]{Mohammad Amir Fallah}
		 received the BSc, MSc, and Ph.D. degrees from Shiraz University, Shiraz, Iran, and Tarbiat Modares University, Tehran, Iran, and Shiraz University, Shiraz, Iran, in 2001, 2003 and 2013 respectively, all in electrical and computer engineering. %After receiving the PhD, he has been working as a project manager in several industrial companies, and 
   He is an assistant professor with the Department of Engineering, Payame Noor University (PNU), Tehran, Iran
, from 2015 till now. His current research interests include antenna and propagation, mobile computing, and the application of machine learning and artificial intelligence in wireless networks.
	\end{biography}

\begin{biography}[{\includegraphics[width=1in,height=1.25in,clip,keepaspectratio]{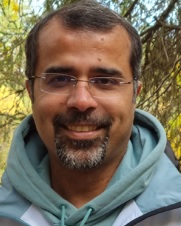}}]{Mehdi Rasti}
		(Senior Member, IEEE) received the
B.Sc. degree in electrical engineering from Shiraz University, Shiraz, Iran, in 2001, and the M.Sc. and Ph.D. degrees from Tarbiat Modares University, Tehran, Iran, in 2003 and 2009, respectively. He is currently an Associate Professor with the Centre
for Wireless Communications, University of Oulu, Finland. From 2012 to 2022, he was with the Department of Computer Engineering, Amirkabir University of Technology, Tehran. From February 2021 to January 2022, he was a Visiting Researcher with the Lappeenranta-Lahti University of Technology, Lappeenranta, Finland.
From November 2007 to November 2008, he was a Visiting Researcher with the Wireless@KTH, Royal Institute of Technology, Stockholm, Sweden. From September 2010 to July 2012, he was with the Shiraz University of
Technology, Shiraz. From June 2013 to August 2013, and from July 2014 to
In August 2014, he was a Visiting Researcher with the Department of Electrical
and Computer Engineering, University of Manitoba, Winnipeg, MB, Canada.
His current research interests include radio resource allocation in IoT, Beyond
5G and 6G wireless networks.
	\end{biography}

 \begin{biography}[{\includegraphics[width=1in,height=1.25in,clip,keepaspectratio]{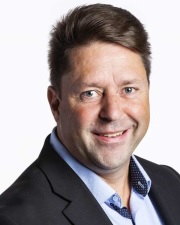}}]{Matti Latva-Aho} (Fellow, IEEE) received the M.Sc., Lic.Tech., and Dr.Tech. (Hons.) degrees in electrical engineering from the University of Oulu, Finland, in 1992, 1996, and 1998, respectively. From 1992 to 1993, he was a Research Engineer at Nokia Mobile Phones, Oulu, Finland, after that he joined the Centre for Wireless Communications (CWC), University of Oulu. He was the Director of CWC, from 1998 to 2006, and the Head of
the Department for Communication Engineering, until August 2014. Currently, he is an Academy of Finland Professor and the Director of the National 6G Flagship Program. He is a Global Fellow with Tokyo University. He has published over 500 journals or conference papers in the field of wireless communications. His group currently focuses on 6G systems research. His research interests include mobile broadband communication systems. In 2015, he received the Nokia Foundation Award for his achievements in mobile communications research.
	\end{biography}
%}

\end{document}